\documentclass[12pt]{article}

\usepackage[a4paper,margin=2.8cm]{geometry}
\usepackage[english]{babel}
\usepackage[utf8]{inputenc}
\usepackage{amssymb,amsmath,amsthm,amsfonts}
\usepackage{algpseudocode}
\usepackage{xspace,xpunctuate}
\usepackage{hyperref}
\usepackage{cleveref}
\usepackage{nicefrac}
\usepackage{mathtools}
\usepackage{comment}
\usepackage{graphicx}

\newbox\brbox
\setbox\brbox\hbox{$[$}
\let\zupp[

\mathcode`\["8000
\def\crazy{\futurelet\next\dostuff}
\def\dostuff{\ifx\next\zupp\schnapp\else\copy\brbox\fi}
{\catcode`\[\active\gdef[{\crazy}}
\def\schnapp#1]]{\copy\brbox\!#1]\!]}

\newtheorem{innercustomthm}{Theorem}
\newenvironment{customthm}[1]
  {\renewcommand\theinnercustomthm{#1}\innercustomthm}
  {\endinnercustomthm}

\newtheorem{innercustomcol}{Corollary}
\newenvironment{customcol}[1]
  {\renewcommand\theinnercustomcol{#1}\innercustomcol}
  {\endinnercustomcol}

\DeclarePairedDelimiter{\floor}{\lfloor}{\rfloor}
\DeclarePairedDelimiter\norm{\lVert}{\rVert}%

\let\oldsim=\sim
\def\sim{\mathbin{\oldsim}}

\newtheorem{proposition}{Proposition}
\newtheorem{lemma}{Lemma}
\newtheorem{theorem}{Theorem}
\newtheorem*{theorem*}{Theorem}
\newtheorem{corollary}{Corollary}
\theoremstyle{definition}
\newtheorem{definition}{Definition}
\newtheorem{observation}{Observation}

\def\fapx{f_\textnormal{apx}}
\def\fid{f_\textnormal{id}}

\def\li{[\![}
\def\ri{]\!]}

\newcommand{\ip}[1]{\li#1\ri}

\newcommand{\N}{\ensuremath{{\mathbf{N}}}\xspace} 
\newcommand{\Z}{\ensuremath{{\mathbf{Z}}}\xspace} 
\def\cal{\mathcal}

\newcommand{\bu}{\bar u \xspace}
\newcommand{\bv}{\bar v \xspace}

\newcommand{\bx}{\bar x \xspace}

\newcommand{\bz}{\bar z \xspace}

\newcommand{\calC}{\mathcal{C}}
\newcommand{\vG}{{\smash{\vec{G}}}}

\newcommand{\fungraph}[1]{\cal G(#1)}
\newcommand{\funformula}[2]{\textnormal{FO}[#1,#2]}
\def\fsig#1{#1_{\it fun}} 

\def\C{{\bf C}}
\def\Nesetril{Ne\v{s}et\v{r}il\xspace}

\renewcommand{\phi}{\varphi}
\renewcommand{\epsilon}{\varepsilon}

\newcommand{\cnt}[1]{\#{#1}\,}

\newcommand{\FOC}{{\normalfont{FOC}}$(\mathbf{P})$\xspace}
\newcommand{\FOCless}{{\normalfont{FOC}}$(\{>\})$\xspace}
\newcommand{\FOCONE}{{\normalfont{FOC}}$_1(\mathbf{P})$\xspace}
\newcommand{\FOX}{{\normalfont{FO}$(\{{>}\kern1pt0\})$}\xspace}
\newcommand{\BOLDFOX}{{FO$\boldsymbol{(\{{>}\,0\})}$}\xspace}

\begin{document}


\title{Approximate Evaluation of\\First-Order Counting Queries%
\thanks{Supported by the German Science Foundation (DFG) under grant
no.~DFG-927/15-1.}
}
\author{Jan Dreier and Peter Rossmanith\\[7pt]
\small Theoretical Computer Science, RWTH Aachen University\\[-2pt]
\small {\tt dreier,rossmani@cs.rwth-aachen.de}}
\date{}
\maketitle

\begin{abstract}
Kuske and Schweikardt introduced the very expressive 
first-order counting logic \FOC to model database queries with counting
operations.  They showed that there is an efficient model-checking
algorithm on graphs with bounded degree, 
while Grohe and Schweikardt showed that probably no such algorithm
exists for trees of bounded depth.

We analyze the fragment \FOX of this logic.
While we remove for example subtraction and comparison between two non-atomic counting
terms, this logic remains quite expressive:
We allow nested counting and comparison
between counting terms and arbitrarily large numbers.
Our main result is an approximation 
scheme of the model-checking problem for \FOX
that runs in linear fpt time on structures with bounded expansion.
This scheme either gives the correct
answer or says ``I do not know.''
The latter answer may only be given if small
perturbations in the number-symbols of the formula
could make it both satisfied and unsatisfied.
This is complemented by showing that
exactly solving the model-checking problem for \FOX is already hard on trees of
bounded depth and just slightly increasing the expressiveness of \FOX
makes even approximation hard on trees.
\end{abstract}


\section{Introduction}

One important task for database systems is to lookup information, which is
usually done in the form of \emph{queries}.  For most modern relational
database management systems, queries are written in the SQL language,
whose logical foundation, the relational calculus, is equivalent to
first-order logic~\cite{Grohe2001}.  This means in particular that every
first-order sentence can be expressed in SQL.

Databases can be represented as relational structures.  The
fundamental problem for first-order logic that corresponds to
the evaluation of a boolean SQL
query in a database is the so called \emph{model-checking problem}:
Given a logical formula $\phi$ and a structure $G$, decide whether $\phi$
is true for $G$, i.e., whether $G$ is a model of $\phi$ (commonly written
$G \models \phi$).
We consider the model-checking problem to be \emph{fixed-parameter
tractable} (fpt) if it can be solved in time $f(|\phi|)\norm{D}^c$
for some function $f$ and constant~$c$ (where $|\phi|$ is the length
of the formula and $\norm{D}$ the size of the database).  Already for
first-order logic, the model-checking problem is AW[$*$]-complete and
therefore unlikely to be fpt, which means that boolean SQL queries are
also hard (and SQL allows other types of queries, too).  In fact, even for
purely existential formulas the model-checking problem is $W[1]$-hard
because finding a $k$-clique is a special case~\cite{DowneyF2013}.
It is therefore natural to ask which classes of structures still admit
fpt model-checking algorithms.  Since relational structures can be
represented by their Gaifman graph (see, e.g., \cite{GroheS2018} for
details of the construction), the above question can be reformulated as
the question for \emph{graph classes} with fpt model-checking algorithms.

\emph{Algorithmic meta-theorems}~\cite{Kreutzer2008} are another motivation for
the model-checking problem.
If a problem can be formulated in
a certain logic then the model-checking algorithm for this logic can
solve it.  Therefore, model-checking results can be seen as meta-theorems
that prove whole families of problems to be algorithmically tractable
on certain classes of inputs.

For graph classes of bounded tree-width Courcelle's theorem states that
the model-checking problem for monadic second-order logic can be solved in time
$f(|\phi|)|G|$ and therefore is fpt~\cite{Courcelle1990}.
Frick and Grohe showed that the dependence on $\phi$ is
non-elementary~\cite{FrickG2004} on specially constructed worst-case
instances, while implementations exist that perform quite well on
``usual'' inputs~\cite{LangerRRS2012}.
However, graph classes with bounded tree-width are very restricted.
It has been shown in a
series of papers that the first-order model-checking problem is efficiently
solvable for more and more sparse graph classes, such as those with bounded degree~\cite{Seese1996},
excluded minor~\cite{FlumG2001},
or locally bounded tree-width~\cite{FrickG2001}, 
culminating in two relatively recent results:  Dvořák, Král$\!$', and
Thomas found a linear fpt algorithm for graph classes of bounded
expansion~\cite{DvorakKT2013} and Grohe, Kreutzer, and
Siebertz~\cite{GroheKS2017} an
algorithm with run time $f(\epsilon,|\phi|)|G|^{1+\epsilon}$
for every $\epsilon>0$ for nowhere dense graph classes.
\Nesetril and Ossona de Mendez introduced bounded expansion and
nowhere dense graph classes generalizing all previously mentioned sparse graph classes.  
They are general enough to capture certain real-world graphs,
as some observations suggest~\cite{DRRVSS2019}.
On the other hand, first-order model-checking is AW[$*$]-complete
on monotone graph classes that are not nowhere dense~\cite{GroheKS2017}.
This makes nowhere dense and bounded expansion graph classes 
two of the most general sparse graph classes 
that still are algorithmically useful.


While it is settled that (at least for monotone graph classes)
we cannot find fpt first-order model-checking algorithms beyond nowhere dense graph classes unless AW[$*$] = FPT,
it is still very much an open question by how much we can extend
first-order logic while keeping the graph classes as general as possible.
This is all the more important since many features of SQL, such as the COUNT operator, 
cannot be properly modeled in first-order logic.
First-order formulas can only make counting-claims of the form
``there are at least $k$ elements with this property'' for some fixed $k \in \N$.
Model-checking results for extensions of first-order logic yield
more general algorithmic meta-theorems, capturing wider ranges of problems.

A limited way to bring counting ability to first-order logic is the 
\emph{query-counting} problem,
where one is given a first-order formula $\phi(\bx)$ with free variables $\bx = x_1\dots x_k$ and a structure $G$
and asked to count the number of tuples $\bv$ of vertices in $G$ such that
$G \models \phi(\bv)$.
This problem is fixed-parameter
tractable on nowhere dense graph classes~\cite{GroheS2018}.
A closely related problem is the \emph{query-enumeration} problem
where one is asked to enumerate satisfying tuples (as in a typical SELECT-statement in SQL).
This problem is tractable on bounded expansion~\cite{KazanaS2013} and
nowhere dense~\cite{SchweikardtSV2018} graph classes.

For more powerful counting mechanisms we are required to extend first-order logic itself.
While many ways have been considered to bring counting to first-order logic~\cite{Vaananen1997,Nurmonen2000,Libkin2004,GroheP2017},
we consider the first-order counting logic \FOC recently introduced by Kuske and Schweikardt~\cite{KuskeS2017}.
In this logic formulas are built according to the rules of first-order logic and from \emph{counting terms}:
A counting term is any number $N \in \Z$ as well as formulas
such as $\#y\,\phi$ standing for ``the number of witnesses for $y$ in $\phi$.''
Counting terms are allowed to be multiplied, added, subtracted and compared
using a collection $\bf P$ of \emph{numerical predicates}.
The precise syntax and semantics can be found in~\cite{KuskeS2017}.
In this work, we restrict ourselves to the binary numerical predicate
$>$ denoting the usual ``greater than'' relation.
The semantics of \FOCless are best illustrated with the help of examples:
The formula
\begin{equation}\label{eq:formula1}
\exists x_1\cdots\exists x_k
\bigl(\#\, y\bigvee_{i=1}^k(x_i=y\lor E(x_i,y)) > N\bigr),
\end{equation}
when evaluated on graphs,
expresses that there are $k$ vertices that dominate more than $N$ vertices.
This describes the \emph{partial dominating set problem}.
The formula
\begin{equation}\label{eq:formula2}
\bigl(\#(x_1,\ldots,x_k)\bigwedge_{i\neq j} E(x_i,x_j)\bigr)
> \bigl(\#(x_1,\ldots,x_k)\bigwedge_{i\neq j} \neg E(x_i,x_j)\bigr)
\end{equation}
expresses that there are more cliques of size~$k$ than
independent sets of size~$k$.
Note that the \emph{length} of each number $N \in \Z$ in a formula is considered to be one.
This means formula (\ref{eq:formula1})
always has constant length and
an fpt model-checking algorithm for \FOC is required to
evaluate it in the same time $f(k)\norm{G}^c$ for any $N$,
even if $N$ depends on $G$.

Kuske and Schweikardt showed that the model-checking problem for \FOC is
fixed parameter tractable for graphs with bounded degree~\cite{KuskeS2017}.
However, already on simple structures, such as trees of bounded depth,
the problem becomes AW[$*$]-complete~\cite{GroheS2018},
and therefore is most likely not fpt.
It seems like the expressive power of \FOC is too strong
to admit efficient model-checking algorithms on more general graph classes.
This invites the question for fragments of \FOC that still admit
efficient model-checking algorithms on 
graph classes with bounded expansion or nowhere dense classes.
In this work, we identify such a fragment.

But let us first mention another fragment (orthogonal to ours), 
introduced by Grohe and Schweikardt~\cite{GroheS2018}.
The fragment \FOCONE is obtained from \FOC by
allowing subformulas of the form $P(t_1,\dots,t_m)$ for some numerical predicate $P \in \bf P$
only if all counting terms $t_1,\dots,t_m$ together contain at most a
single free variable.
The above formula (\ref{eq:formula2}) is in FOC$_1(\{>\})$,
since both counting terms have zero free variables.
However, the formula (\ref{eq:formula1}) for partial dominating set is not in FOC$_1(\{>\})$ (unless $k=1$)
because the counting term $\#y\,\bigvee_{i=1}^k(x_i=y\lor E(x_i,y))$
has $k$ free variables.  
Grohe and Schweikardt showed that the model-checking problem for
\FOCONE is fixed parameter tractable on nowhere dense graph classes~\cite{GroheS2018}.
For bounded expansion graph classes,
Toru\'nczyk presents an even stronger query language (also orthogonal to ours) that extends first-order logic
by aggregation in multiple semirings~\cite{Torunczyk20}.

\paragraph{Results.}
In this work, we consider the fragment of \FOCless
built recursively using the rules of first-order logic and the following rule:
\begin{quote}
    If $\phi$ is a formula, $y$ is a variable, and $N \in \Z$,
    then $\cnt{y}\phi > N$ is a formula.
\end{quote}
Except for syntactic differences, this fragment is equivalent to the logic \FOX defined by Kuske and Schweikardt~\cite{KuskeS2017},
where ``${>}\,0$'' stands for the unary predicate testing whether a term is positive.
In their definition of \FOX one may write $\cnt{y}\phi - N > 0$ instead of $\cnt{y}\phi > N$.
To avoid having multiple names for the same logic and because all 
our results are independent of such syntactic differences,
we call our fragment \FOX as well.
This logic further exists under the name FO(C)~\cite{EbbinghausF2005}


Formula (\ref{eq:formula1}) for partial dominating set is in \FOX
while (\ref{eq:formula2}) compares two non-constant counting terms
and therefore is \emph{not} a \FOX formula.
This makes \FOX and FOC$_1(\{>\})$ incomparable.
But while model-checking for \FOCONE is fixed parameter trac\-ta\-ble on nowhere dense graph classes,
\FOX is still too expressive for efficient model-checking: 
We prove that, just like \FOC, the model-checking problem for \FOX is AW[$*$]-hard even on trees of bounded depth (\Cref{lem:exacthard}).

For this reason, we define the concept of \emph{approximate model-checking}.
An approximate model-checking algorithm gets as input a graph $G$,
a formula $\phi$ and an accuracy $\varepsilon > 0$,
runs in time $f(|\phi|,\varepsilon)|G|$,
and either returns 
$1$ (meaning $G \models \phi$),
$0$ (meaning $G \not\models \phi$), or
$\bot$ (meaning ``I do not know.'')
The symbol $\bot$ may only be returned if slight perturbations
in the constants of $\phi$ could make the formula either satisfied or unsatisfied.
For smaller values of $\varepsilon$, these perturbations need to be increasingly small.
Our main result is the following:
\begin{theorem}\label{thm:approxScheme}
There is a linear fpt model-checking approximation scheme for
\FOX on labeled graph classes with bounded expansion.
\end{theorem}
This means for every graph class with bounded expansion
there exists a function $f$ such that the model-checking problem for \FOX 
on this graph class can be approximated with an arbitrary accuracy $\varepsilon > 0$ in time $f(|\phi|,\varepsilon)|G|$
(\Cref{def:modelcheckingblub}).

Let us now describe when the approximation algorithm is allowed to answer $\bot$.
For $\lambda>1$ we call two formulas \emph{$\lambda$-similar} if one formula can be obtained from the
other one by changing the constant counting terms by a factor between
$1/\lambda$ and~$\lambda$.  The two \FOX-formulas 
$$
\#y(\#z\,\phi(yz)<500) >1000
$$
$$
\#y(\#z\,\phi(yz)<498) >1009
$$
are $1.01$-similar, but not $1.009$-similar.
We further say a formula $\phi$ is $\lambda$-\emph{unstable}
on a graph $G$ if $\phi$ is $\lambda$-similar to two formulas
$\phi'$ and $\phi''$ such that $G\models\phi'$ and
$G\not\models\phi''$.
For a given $\varepsilon > 0$, the approximation algorithm
is only allowed to answer $\bot$ if the input formula $\phi$
is $(1+\varepsilon)$-unstable on the input graph $G$.
Note that formulas without counting quantifiers are never unstable
and may never lead to the answer $\bot$.
Our approximation scheme therefore generalizes the first-order model-checking problem.

It can be argued that answering queries 
approximately is in many applications almost as good as an exact answer
because the involved numbers (like a maximal debt of one million dollars
or a maximum allowed temperature of 1000 degrees) are often only
ballpark numbers.
Furthermore, if $\bot$ is returned we know that the formula is ``close''
to being satisfied and unsatisfied, 
which sometimes may be interesting in itself.
For example, if the partial dominating set formula (\ref{eq:formula1})
is $(1+\varepsilon)$-unstable then
there exists a solution dominating more than $N/(1+\varepsilon)$
vertices, but none dominating more than $(1+\varepsilon)N$ vertices.

A natural question that arises is whether \FOX can
be generalized while keeping the ability of efficient approximate
model-checking.  
We answer this question negatively.
If we make \FOX just slightly more powerful then
approximate model-checking becomes $\rm AW[*]$-hard
even on trees of bounded depth.
This happens if we allow either counting
quantification on pairs of variables, subtraction, multiplication, or
comparison between two non-constant counting terms,
and even for very large ``approximation ratios''
(\Cref{lem:approxhard}).
The \FOX fragment seems to be at the edge of
what can still be efficiently approximated on bounded expansion graph classes.

Nevertheless, we identify certain \FOX-formulas that can be evaluated exactly.
They are of the form $\exists x_1\ldots\exists x_k\#y\,\phi(y\bx)$
where $\phi(y\bx)$ is a first-order formula.
Since the previously mentioned partial dominating set formula (\ref{eq:formula1})
is of this shape, we get the following result:

\begin{corollary}\label{col:partialDomSet}
Partial dominating set can be solved in linear fpt time on graph classes with
bounded expansion.  
\end{corollary}

It has been shown by Amini, Fomin, and Saurabh~\cite{AminiFS11}, 
that partial dominating set can be solved in fpt
time on minor-closed graph classes,
but the complexity on graph classes with
bounded expansion has remained open.  Moreover, the running time
of Amini et al.'s algorithm for an $H$-minor free graph class is
of the form $f(k)n^{c_H}$, where $n$ is the number of vertices and
$c_H$ is a constant that depends on~$H$, while our running time 
is linear in~$n$.  We get the same running time for similar problems such as
distance-$r$ dominating set or variants of partial vertex cover.
We are further able to solve a general optimization problem
where the goal is to retrieve an optimal witnesses for the free variables of a counting formula.

\begin{theorem}\label{thm:optimization}
Let $\cal C$ be a labeled graph class with bounded expansion.
There exists a function $f$ such that for a given
graph $G \in \cal C$ and first-order formula $\phi(y\bx)$
one can compute in time
$f(|\phi|)\norm{G}$ a tuple
$\bu^* \in V(G)^{|\bx|}$ such that
$$
[[\#y\, \phi(y\bu^*)]]^{G} = 
\underset{\bu \in V(G)^{|\bx|}}{\textnormal{opt}} [[\#y\, \phi(y\bu)]]^{G},
$$
where \textnormal{opt} is either $\min$ or $\max$.
\end{theorem}

This means, for example, that we can find an optimal partial dominating
set of size $k$ in linear fpt time, which is faster than
using self-reducibility.
Another example is MaxSAT~\cite{ChaIKM1997,MillsT2000}, parameterized
by the search depth $k$ of a local search.
Assume we have a SAT formula whose incidence graph comes from a class
with bounded expansion and an (unsatisfying) assignment.
We could use \Cref{thm:optimization} to find in linear fpt time
another assignment with Hamming distance at most $k$ maximizing
the number of satisfied clauses.

\paragraph{Techniques.} Most of our
proofs use functional structures to represent graphs.
The overall strategy of our main result in \Cref{thm:approxScheme}
is the use of quantifier elimination to
replace the model-checking problem by one with one counting quantifier
less until we reach a quantifier-free formula. 
To eliminate a counting quantifier, we perform a sequence of transformations
on counting terms of the form $\#y\,\phi(y\bx)$ where $\phi(y\bx)$ is quantifier-free.
We replace them with a sum of gradually simpler counting terms
until they are simple enough to be directly evaluated.
Most transformations preserve the value of the counting term.
In the end, however, we have to replace each summand with an approximation of it.
This leads to a problem at one point.
We can express a counting term via inclusion-exclusion
as $a-b$ for two terms $a$ and $b$. 
If we have an approximation $a'$ of $a$ and $b'$ of $b$ with good relative error
and $a=b$ are very large then $a'-b'$ may be a very bad approximation of $a-b=0$.
This has to be avoided.
If $b$ is rather small we can ignore it
and just use $a$ as a good approximation of $a-b$. 
If $b$ is big, however, something needs to be done.
By a preprocessing of the graph during which we add so-called ``flip'' arcs we modify it in
such a way that the subtraction $a-b$ can be done exactly whenever necessary.
This is the most crucial step in the proof as we have to be very
careful to add enough arcs to achieve the necessary precision, while still
staying in a graph class with bounded expansion. 
At last, we have approximated the counting term $\#y\,\phi(y\bx)$
using a sum of simpler counting terms of the form $\#y\,\psi(yx_i)$.
Each simpler counting term depends only on one free variable
and can therefore be evaluated in linear fpt time.
Then we
round the resulting numbers into a constant number of intervals and introduce
unary predicates indicating the intervals in which the numbers lie.  This is
the second situation where we loose precision.  Using these predicates we
can finally get rid of the counting quantifier and replace
a subformula of the form $\#y\,\phi(y\bx) > N$
with a quantifier-free one.
Due to the previously introduced errors,
the new formula may not always give us the correct answer.
Therefore we build a pair of quantifier-free formulas: one over- and one underapproximation.
If they agree, we know the correct answer.  
If they disagree, we can be sure that the situation is unstable.

If the term $\#y\,\phi(y\bx)$ is not part of a larger formula, we 
we can stop the quantifier elimination step early.
We avoid the problem of subtraction and the encoding into unary predicates
and instead evaluate the simplified intermediate
counting terms directly using standard methods.
This means that we can solve the model-checking problem for such
formulas \emph{exactly} in linear fpt time,
giving rise to \Cref{thm:optimization}.

The remaining paper is structured as follows:  We start by introducing
the necessary notation.  Then in \Cref{sec:approxModelChecking} we develop the machinery for quantifier
elimination in functional representations of graphs.  We then 
prove the main result that there is an approximate model-checking
algorithm for \FOX on graph classes with bounded expansion (\Cref{thm:approxScheme}).
In Section~\ref{sec:exact} we prove 
our exact model-checking and optimization algorithm for
\FOX-formulas of a special shape (\Cref{thm:optimization}).
At last, in \Cref{sec:hardness} we prove the hardness of
exact model-checking for \FOX and approximate model-checking for
generalizations of \FOX (Lemma \ref{lem:exacthard}, \ref{lem:approxhard}).

\section{Definitions and Notation}

\paragraph*{Graphs.}

In this work we obtain results for \emph{labeled graphs}.
A labeled graph is a tuple $G=(V,E,P_1,\dots,P_m)$,
where $V$ is the vertex set, $E$ is the edge set
and $P_1,\dots,P_m \subseteq V$ the labels of $G$.
The \emph{order} $|G|$ of $G$ equals $|V|$.
We define the \emph{size} $\norm{G}$ of $G$ as $|V|+|E|+|P_1|+\dots+|P_m|$.
Unless otherwise noted, our graphs are undirected.
For a directed graph $G$, the indegree of a node $v$ equals the number of
vertices $u$ such that there is an arc $uv$ in $G$.
The maximal indegree of all nodes in $G$ is denoted by $\Delta^-(G)$.

\paragraph*{Logic.}

We consider fragments of the very general first-order counting
logic \FOC defined by Kuske and Schweikardt~\cite{KuskeS2017}.
It depends on a collection $\bf P$ of numerical predicates,
where each predicate $P \in \bf P$ has semantics specified by $[[P]]$.
We consider fragments of \FOCless 
where $>$ is the usual ``greater than'' predicate with
$[[{>}]] = \{\,(a,b) \in \Z^2 \mid a > b\,\}$.
As we will only use a subset of \FOCless, we refrain from
giving the whole definition.
Instead, we define fragments of \FOCless, as we introduce them.
The semantics of \FOC are as expected
and we refer the reader to~\cite{KuskeS2017} for a rigorous definition.

\begin{definition}\label{def:fox}
We define \FOX to be the fragment of \FOCless built
using the rules of first-order logic (rule 1,2,3 in \cite[Definition 2.1]{KuskeS2017}) and the following rule:
\begin{quote}
    If $\phi$ is a formula, $y$ is a variable, and $N \in \Z$,
    then $\cnt{y}\phi > N$ is a formula.
\end{quote}
\end{definition}
Except for syntactic differences, this definition
is equivalent to the original definition of \FOX provided by Kuske and Schweikardt~\cite{KuskeS2017}
(with the original syntax one has to write $\cnt{y}\phi - N > 0$ instead of $\cnt{y}\phi > N$).

We say a \FOC formula is \emph{quantifier-free} if it contains
no $\exists$, $\forall$, or $\#$ quantifiers.
If two formulas $\phi_1$ and $\phi_2$ are logically equivalent we write $\phi_1 \equiv \phi_2$.
The \emph{length} of a formula $\phi$ is denoted by $|\phi|$ and equals its number of symbols.
In particular, the length of any number-symbol $N \in \Z$ in a \FOC formula is
one (and should not be confused with the length of a binary encoding of $N$).
For two signatures we write $\sigma \subseteq \rho$ to indicate that $\rho$ extends $\sigma$.
All signatures are finite and the cardinality $|\sigma|$ of a signature equals its number of symbols.
We often interpret a conjunctive clause $\omega \in$ FO as a set of literals and write
$l \in \omega$ to indicate that $l$ is a literal of $\omega$.

We denote the universe of a structure $G$ by $V(G)$.
We interpret a labeled graph
$G = (V,E,P_1,\dots,P_m)$ as a logical structure with universe
$V$, binary relation $E$ and unary relations $P_1$, \dots, $P_m$.

The notation $\bx$ stands for a non-empty tuple $x_1\dots x_{|\bx|}$.
We write $\phi(\bx)$ to indicate that a formula $\phi$
has free variables $\bx$.
Let $G$ be a structure, $\bu \in V(G)^{|\bx|}$ be a tuple of elements from the universe of $G$,
and $\beta$ be the assignment with $\beta(x_i) = u_i$ for $i \in \{1,\dots,|\bx|\}$.
For simplicity, we write $G \models \phi(\bu)$ and $[[\phi(\bu)]]^G$
instead of $(G,\beta) \models \phi(\bx)$ and $[[\phi(\bx)]]^{(G,\beta)}$.

Further notation concerned with functional structures and formulas
is introduced in \Cref{sec:functionalRepresentations}.

\paragraph*{Model-Checking.}
Let $\cal C$ be a class of structures and L be a logic.
The \emph{parameterized model-checking problem for \textnormal{L} on $\cal C$} is 
the defined as follows: 
The input is a structure $G \in \cal C$ and a sentence $\phi \in$ L
with matching signatures. The parameter is $|\phi|$.
The question is whether $G \models \phi$.
The parameterized first-order model-checking problem on the class
of all graphs is a complete problem for the complexity class AW[$*$].
As the whole W-hierarchy is contained in AW[$*$]
it is generally assumed that $\rm AW[*]\not\subseteq FPT$.

\paragraph*{Model-Checking Approximation Scheme.}
We now define our novel notion of a model-checking approximation scheme,
an fpt algorithm which is only allowed to answer ``I do not know''
if the fact whether the structure is a model of the formula
is sensitive to slight perturbations in the constants of the formula.


\begin{definition}[$\lambda$-similarity]\label{def:similar}
    Let $\lambda > 1$ and $\phi$ be a \FOC formula.
    A \FOC formula $\phi'$ is \emph{$\lambda$-similar} to $\phi$
    if $\phi'$ can be obtained from $\phi$ by replacing
    each atomic counting term $t \in \Z$ of $\phi$ by $t' \in \Z$ with
    $t/\lambda \le t' \le \lambda t$.
\end{definition}

\begin{definition}[$\lambda$-stability]\label{def:stable}
    Let $\lambda > 1$, $G$ be a structure
    and $\phi$ be a \FOC sentence.
    We say $\phi$ is \emph{$\lambda$-stable} on $G$
    if for every \FOC sentence $\phi'$ which is $\lambda$-similar to $\phi$
    it holds that $G \models \phi$ iff $G \models \phi'$.
    Otherwise we say that $\phi$ is \emph{$\lambda$-unstable} on~$G$.
\end{definition}

\begin{definition}[linear fpt model-checking approximation scheme]\label{def:modelcheckingblub}
    Let $\cal C$ be a class of labeled graphs,
    and \textnormal{L} be a fragment of \FOC.
    A \emph{linear fpt model-checking approximation scheme} for the logic \textnormal{L} on the class $\cal C$
    is an algorithm that gets as input a sentence $\phi \in \textnormal{L}$, a graph $G \in \cal C$ and $\epsilon > 0$,
    runs in time at most $f(|\phi|,\epsilon)\norm{G}$ for some function $f$
    and returns either 1, 0, or $\bot$.
    \begin{itemize}
    \item
        If the algorithm returns 1 then $G \models \phi$.
    \item
        If the algorithm returns 0 then $G \not\models \phi$.
    \item
        If the algorithm returns $\bot$ then $\phi$ is $(1+\epsilon)$-unstable on $G$.
    \end{itemize}
\end{definition}

\section{Approximate Model-Checking}\label{sec:approxModelChecking}

We will work with graph classes with bounded expansion and use their
characterization via transitive fraternal augmentations (\Cref{sec:flip}).  An
undirected graph is first replaced with a directed graph by orienting the
edges in such a way that the indegree is bounded by a constant that
depends only on the graph class.  We represent this directed graph by
a functional structure (\Cref{sec:functionalRepresentations}).
The signature of this structure
consists of a constant number of function symbols (usually denoted by $f,g,h$) and 
unary predicate symbols.  The function symbols represent arcs.
If $f(u)=v$ for some function $f$, then the corresponding directed graph has an arc~$vu$. 
In this way, we need only as many function symbols as the indegree of the directed graph.

Our model-checking algorithm works via \emph{quantifier elimination}.
This means, we gradually simplify the input formula
by iteratively removing the innermost quantifier.
We compensate every removed quantifier by adding new arcs and unary relations to our input structure (maintaining bounded expansion).
When no quantifiers are left we can easily evaluate the formula.
In this procedure, the subformulas spanned by the innermost quantifier
are of the form $\#y\, \phi(y\bx) > N$ (where $\phi$ is quantifier-free).
We want to replace such a formula with two almost equivalent quantifier-free formulas (\Cref{sec:iteratedElimination}).
If evaluating these formulas on a graph gives two different results
then we know 
that $\#y\, \phi(y\bx)$ is up to a factor of $(1+\varepsilon)$ close to $N$
and we are allowed to return~$\bot$.

In \Cref{sec:positiveSums}
we gradually transform the innermost counting term $\#y\, \phi(y\bx)$
into simpler terms while expanding the corresponding functional structure with
new arcs and unary relations.
In the end, we obtain
a sum of $t$ simpler terms of the form $\#y\,\tau(y)\land \psi(\bx) \land f(y)=g(x_i)$.
If we ignore $\psi(\bx)$, the simpler terms only have a single free variable
and can be evaluated in linear time for all inputs.
We divide the numbers from $0$ to $N$ into $t/\varepsilon$ buckets
and introduce unary predicates $R_0,\dots,R_{t/\varepsilon}$,  where
$R_l(x_i)$ is true if and only if the value of $\#y\,\tau(y)\land f(y)=g(x_i)$ is in the $l$th bucket.
Thus if for each summand and each $l$ we know the value of $\psi(\bx)$ and $R_l(x_i)$,
we either know that $\#y\, \phi(y\bx)$ is greater than $N$ or can reconstruct its value of up to a factor of $(1+\varepsilon)$ (\Cref{sec:constructingFO}).
This reconstruction can be done in a quantifier-free first-order formula with free variables $\bx$,
which completes the quantifier elimination.
Note that summands are not allowed to be negative since
due to cancellation the magnitude of individual summands could be 
considerably larger than the final sum
and the bucket-rounding technique would not work.

The main challenge is to find a decomposition of $\#y\, \phi(y\bx)$
into summands of the form $\#y\,\tau(y)\land \psi(\bx) \land f(y)=g(x_i)$.
This transformation is done in several stages.
The first intermediate step are formulas that consist of 
conjunctive clauses that can be grouped as
$\tau(y)\land\psi(\bx)\land\Delta^=(y\bx)\land\Delta^{\neq}(y\bx)$
and is carried out in a similar way to what Kazana and Segoufin
did~\cite{KazanaS2013}.  Here all $\tau(y)$, $\psi(\bx)$, $\Delta^=(y\bx)$,
and $\Delta^{\neq}(y\bx)$ are conjunctions of atomic formulas, which we also call \emph{literals}.
Those literals that contain only $y$ are grouped into $\tau(y)$, those with variables only from
$\bx$ into~$\psi(\bx)$.  
We call the remaining ones the \emph{mixed literals}, as they depend
on $y$ and $\bx$. We make sure that they are either of the form
$f(y)=g(x_i)$ or $f(y) \neq g(x_i)$.
The former ones are placed into $\Delta^=(y\bx)$ 
and the latter ones into $\Delta^{\neq}(y\bx)$.

Let us replace the input graph with its 1-transitive fraternal augmentation.
By the fraternal rule,
for every original function symbol $f$, $f'$ 
and every vertex $v$ there is a function symbol 
$h$ in the transitive fraternal augmentation 
with either $f(v) = h(f'(v))$ or $h(f(v)) = f'(v)$.
By expanding each conjunctive clause, we take care that $\tau(y)$ contains every possible literal
of the form $h(f(y))=g(y)$ or its negation.
Similarly for literals $h(f(x_i))=g(x_j)$ in $\psi(\bx)$.
This creates redundancy that helps us replace mixed literals.
We will proceed in a similar way as Kazana and Segoufin, but have to be a bit
more
careful about not overcounting, as we eliminate a counting quantifier rather
than an existential one.
As the next step, we make sure that $\Delta^=(y\bx)$ contains only one literal.
The resulting formulas
$\tau(y)\land\psi(\bx)\land f(y) = g(x_i) \land\Delta^{\neq}(y\bx)$
are one step closer to their final form.
It remains to eliminate the negative literals in $\Delta^{\neq}(y\bx)$.
A standard way to do so would be inclusion-exclusion.
But this is not allowed since it would lead to subtraction, which cannot be approximated.
Finally, we have to use very different techniques than Kazana and Segoufin~\cite{KazanaS2013}.

If the negative mixed literals 
in $\Delta^{\neq}(y\bx)$ are satisfied by almost all of
the witnesses for $y$, 
then removing them increases the final count only a little bit.
This way we obtain a good enough approximation.
If this does not work, we introduce so called ``flip'' arcs (see \Cref{sec:flip}) to the graph 
and exploit the redundant literals added to $\tau(y)$ and $\psi(\bx)$.
The redundancy in $\tau(y)\land\psi(\bx)\land f(y) = g(x_i)$
together with the new arcs then
imply $\Delta^{\neq}(y\bx)$ to be either always true or always false.
In the former case, we remove $\Delta^{\neq}(y\bx)$ from its
conjunctive clause, and 
in the latter case we can remove the whole conjunctive clause.
Our key observation is that 
we only need to introduce a small amount of ``flip'' arcs
and therefore stay within a graph class with bounded expansion.

\subsection{Transitive Fraternal Flip Augmentations}\label{sec:flip}

A directed graph $G'$ is a \emph{$1$-transitive fraternal augmentation}
of a directed graph $G$ if it has the same vertex set as $G$ and satisfies the following
conditions~\cite{NesetrilM2012}:
\begin{itemize}
\item {\bf Transitivity.} If the arcs $uv$ and $vw$ are present in $G$ then
$uw$ is present in $G'$.
\item {\bf Fraternity.} If $uw$ and $vw$ are present in $G$ then
$uv$ or $vu$ are present in $G'$.
\item {\bf Tightness.} If $G'$ contains an arc that is not present in $G$
it must have been added by one of the previous two rules.
\end{itemize}

Let $G$ be an undirected graph.
We call a sequence $G_0\subseteq G_1\subseteq\cdots$ a
\emph{transitive fraternal augmentation}
of $G$ if $G_0$ is a directed graph obtained by orienting the edges of~$G$,
and $G_{i+1}$ is a $1$-transitive fraternal augmentation of $G_i$ for $i \ge 0$.
For any graph class $\cal C$ with bounded expansion 
\Nesetril and Ossona de Mendez devised an algorithm~\cite{NesetrilM2008b} that
computes a transitive fraternal augmentation
$G_0\subseteq G_1\subseteq G_2\subseteq\cdots$ of $G$
such that $\Delta^-(G_i)\leq\Gamma_{\cal C}(i)$, $i \in \N$
for a function $\Gamma_{\cal C}$ that depends only on the graph class~$\cal C$.
The orientation $G_0$ can be computed in time $O(\norm{G})$ from $G$,
and $G_{i+1}$ can be computed from $G_i$ in time $O(\norm{G_i})$.
We will assume that this algorithm is used to compute augmentations and orientations
and call the corresponding output \emph{the}
$1$-transitive fraternal augmentation and \emph{the} orientation,
similarly to what Kazana and Segoufin did~\cite{KazanaS2013}.  

We will need a generalization of this construction.
We say a directed graph $G'$ is a flip
of $G$ if $G'$ is a supergraph of $G$ with the same vertex set
that can contain additional ``flipped'' arcs, i.e.,
$G'$ can have an additional arc $uv$ only if $G$ already
contains~$vu$.  Moreover, we require that
$\Delta^-(G')\leq\Delta^-(G)+1$.  In a flip, with other words,
we can add reverse arcs for arcs that are already present without
increasing the maximal indegree by more than one.
A sequence $G_0 \subseteq G_1 \subseteq G_2 \subseteq \cdots$ 
is called a \emph{transitive fraternal flip augmentation} of $G$
if $G_{i+1}$ is a flip or the
$1$-transitive fraternal augmentation of $G_i$ for all $i\ge0$
and $G_0$ is the orientation of~$G$.
Note that we can apply \emph{any}
flip but can only apply the ``well-behaved'' orientation and augmentation
devised by \Nesetril and Ossona de Mendez~\cite{NesetrilM2008b}.
We can characterize graph classes with bounded expansion via transitive
fraternal flip augmentations.

\begin{lemma}\label{lem:flip}
Let $\calC$ be a graph class.  Then $\calC$ has bounded expansion if and
only if there exists a function $\Gamma_{\calC} \colon \mathbf{N} \to
\mathbf{N}$ such that every graph $G \in \calC$ and every transitive
fraternal flip augmentation $G_0 \subseteq G_1 \subseteq G_2 \subseteq \dots$ 
of $G$ has $\Delta^-(G_i) \le \Gamma_{\calC}(i)$ for every
$i\geq0$.
\end{lemma}
\begin{proof}
Assume $\cal C$ has bounded expansion.
The maximal density of an $r$-shallow minor of a graph $G$ is denoted
by $\nabla_r(G)$.  
When we apply $\nabla_r$ to a directed graph we mean
$\nabla_r$ of the underlying undirected graph.  The exact
definition of $\nabla_r$ is not important for this proof as we will
use it as a black box.
By definition of bounded expansion, $\nabla_r(G)$ is bounded by some function
of $r$ for all graphs $G \in \cal C$~\cite{NesetrilM2008a}.  
The construction of \Nesetril and Ossona
de Mendez provides an orientation $G_0$ of $G$ such that
$\Delta^-(G_0)\leq2\nabla_0(G)$ \cite[Fact~3.1]{NesetrilM2008a}.
Also $\nabla_r(G_0) = \nabla_r(G)$.
Thus, if $G \in \cal C$
then the values $\Delta^-(G_0), \nabla_{0}(G_0), \nabla_{1}(G_0), \dots$ can all be bounded independent of $G$.

Let $G_{i+1}$ be the $1$-transitive fraternal augmentation of ~$G_i$.
\Nesetril and Ossona de Mendez showed that $\nabla_r(G_{i+1})\leq
p_{2r+1}(\Delta^-(G_i)+1,\nabla_{2r+1}(G_i))$ for some polynomial
$p_{2r+1}$~\cite[Lemma~3.5]{NesetrilM2008b}.  Also $\Delta^-(G_{i+1})
\le \Delta^-(G_i)^2 + 2\lfloor \nabla_0(G_i) \rfloor$~\cite[Chapter
4.1]{NesetrilM2008b}.
If $G_{i+1}$ is a flip of $G_i$
then $\Delta^-(G_{i+1})\leq\Delta^-(G_i)+1$ and
$\nabla_r(G_{i+1})=\nabla_r(G_i)$
for all~$r$ because the underlying undirected graphs are the same.
Hence, if $\Delta^-(G_i), \nabla_{0}(G_i), \nabla_{1}(G_i), \dots$ are bounded,
so are $\Delta^-(G_{i+1}), \nabla_{0}(G_{i+1}), \nabla_{1}(G_{i+1}), \dots$.
By induction this gives us a function  $\Gamma_{\calC}(i)$ 
with $\Delta^-(G_i) \le \Gamma_{\calC}(i)$ for every $i \ge 0$.
The other direction follows directly from
\cite[Corollary~5.3]{NesetrilM2008a}.
\end{proof}

\subsection{Functional Representations}\label{sec:functionalRepresentations}

We prove our results using a functional representation of graphs.
They were used heavily by Durand and Grandjean~\cite{DurandG2007}
and again by Kazana and Segoufin~\cite{KazanaS2013}, but partially
also by Dvořák, Král$\!$', and Thomas in the first proof that
first-order model-checking is ftp on bounded expansion graph
classes~\cite{DvorakKT2013}.  One big advantage of functional
representations is the ability to talk about short paths without using
quantifiers as long as all indegrees are bounded.

A \emph{functional signature} is a finite signature containing
functional symbols of arity one and unary predicates.
For a functional signature $\sigma$, we will denote the set of function symbols
by $\fsig\sigma$.
A \emph{functional representation} of a labeled directed graph $G$
is a $\sigma$-structure $\vG$.
The universe of $\vG$ is~$V(G)$.
For every label of $G$ there is one unary predicate in $\vG$ representing it.
The arcs of $G$ are represented using functions.
An arc $uv$ is present in $G$ if and only if $f^{\vG}(v) = u$
for some function symbol $f \in \fsig\sigma$.
Note that we need only $\Delta^-(G)$ different function symbols.
Unused function symbols are mapped to the vertex itself,
in particular an isolated vertex $v$ has $f^{\vG}(v) = v$
for every $f$ in $\fsig\sigma$.
For solely technical reasons, we further require a special
function symbol $\fid$ where $f_{\text{id}}^{\vG}(v) = v$
for all $v \in V(\vG)$.
We call $G$ the \emph{underlying directed graph} of $\vG$ and
by the \emph{underlying undirected graph} of $\vG$
we mean the underlying undirected graph of $G$.

We define the \emph{size} $\norm{\vG}$ of $\vG$ as $|\vG||\sigma|$ 
where $\sigma$ is the signature of $\vG$.
We transfer all remaining notation from directed graphs to functional
representations as expected.
For example $\Delta^-(\vG)$ is defined as $\Delta^-(G)$.
For a given functional signature $\sigma$
we define $\fungraph{\sigma}$ to be the class of all functional representations
with signature~$\sigma$.

Our logics FO or \FOC are defined in the usual way for this functional setting.
Note that in particular we allow nested function terms such as $f(g(x))$.
The \emph{functional depth} of a formula is the maximum number of nested function applications.
For example $f(g(x))=y$ has functional depth $2$.
We define $\funformula{d}{\sigma}$ to be all first-order formulas with
functional signature $\sigma$ and functional depth at most $d$.

For a given graph $G$ we later want
to have a sequence of functional representations
$\vG_0 \subseteq \vG_1 \subseteq \dots$ 
such that the sequence of underlying directed graphs
$G_0 \subseteq G_1 \subseteq \dots$ forms a transitive fraternal augmentation
of $G$ and additionally $\vG_{i+1}$ is an expansion of $\vG_i$ for $i \ge 0$.
We will later heavily exploit that for every
sentence $\phi$ with the same signature as $\vG_0$ and $i > 0$ it
holds that $\vG_0 \models \phi$ iff $\vG_i \models \phi$.

We extend our notion of $1$-transitive fraternal augmentations and flips
to functional representations.
For a given functional representation $\vG$,
we obtain the \emph{$1$-transitive fraternal augmentation} $\vG'$ of
$\vG$ by adding new function symbols to $\vG$ representing
all newly introduced arcs.
The functions representing the transitive edges are added in a special way:
For all function symbols $f,g$ in the signature of $\vG$ 
we add a function symbol $h_{f,g}$ to the signature of $\vG'$ and
define $h_{f,g}^{\vG'}= g^{\vG'}\circ f^{\vG'}$ representing
the newly introduced transitive edges obtained from $f$ and~$g$.
This step will later help us simplify our formulas by replacing 
nested functions of the form $g(f(x))$ with a single function $h_{f,g}(x)$.
Fraternal edges are added as well, of course, but we do not require
any special naming for them.
The construction of $\vG'$ is not necessarily deterministic,
but we can assume it to be.
Note that if $\Delta^-(\vG')$ is bounded then 
the signature of $\vG'$ has only a constant number of new function symbols.

A \emph{flip} $\vG'$ of $\vG$ is an expansion of $\vG$ with the same universe
and one more functional symbol representing the flipped edges.
Since the indegree of a flip may increase by at most one, one new function
symbol is sufficient.
At last, we define what it means for a class of functional representations to
have bounded expansion.

\begin{definition}\label{def:be-functional}
We say a class $\cal C$ of functional representations has bounded
expansion if there exists a functional signature $\sigma$ such that
$\cal C \subseteq \fungraph{\sigma}$
and the class of all underlying undirected graphs has bounded expansion.
\end{definition}

\begin{corollary}\label{cor:funcBE}
Let $\cal C$ be a class of functional representations with bounded expansion.
Consider the class $\cal C'$ of all functional representations $\vG'$
such that $\vG'$ is either the $1$-transitive fraternal augmentation
or a flip of $\vG$ for some $\vG \in \cal C$.
Then $\cal C'$ has bounded expansion.
\end{corollary}

Let $\cal C$ be a graph class with bounded expansion.
For a graph $G \in \cal C$ we can compute
a functional representation $\vG$
of the orientation of $G$ in time $O(\norm{\vG})$.
Let us assume $\vG$ has signature $\sigma$.
The functional formula 
$
\eta(x,y) = \bigvee_{f \in \fsig\sigma} f(x) = y \lor f(y) = x
$
is true for some pair of vertices in $\vG$
if and only if there is an edge between them in~$G$.
Instead of evaluating some relational formula on~$G$,
we can replace every edge relation $E(x,y)$ with $\eta(x,y)$
and evaluate the resulting functional formula on~$\vG$.

Similar to Kazana and Segoufin~\cite{KazanaS2013},
we restrict ourselves to finding algorithms
for classes of functional representations with bounded expansion.
As discussed above (and in \cite{KazanaS2013}),
they also work for graph classes with bounded expansion.

Most of the time we will be using functional representations.
To be less verbose (and when it is clear from the context),
we will call functional signatures simply \emph{signatures},
classes of functional representations simply \emph{classes}.

\subsection{Approximating Counting Terms using Positive Sums}\label{sec:positiveSums}

We will often deal with formulas in disjunctive normal form, i.e.,
a disjunction of conjunctions of literals.  We will call the conjuncts
often \emph{conjunctive clauses} and sometimes only \emph{clauses}
when the exact meaning is clear form the context.  An important
technical tool in the upcoming proofs are special forms of conjunctive
clauses that will be defined next.  
They are partially \emph{complete} in the sense that they must contain certain
atomic formulas or their negation.  This completeness will force the
value of other literals and allow us to remove them from the
conjunctive clause, which is one simplification step of many more to
come.

\begin{definition}\label{def:canonical}
Let $\sigma$, $\rho$ be signatures with $\sigma \subseteq \rho$.
A conjunctive clause
$\tau(y)\land\psi(\bx)\land \Delta^=(y\bx)\land\Delta^{\neq}(y\bx)$
with $\bx = x_1,\dots,x_k$ is called a 
\emph{$k$-$\sigma$-$\rho$-canonical conjunctive clause} if
\begin{enumerate}
    \item $\tau(y) \in \funformula{2}{\rho}$ 
    is a conjunctive clause that
    contains for every $f,g,h, \in \fsig\rho$,
    either the literal $f(y) = h(g(y))$ or its negation.
\item $\psi(\bx) \in \funformula{2}{\rho}$ 
    is a conjunctive clause that
    contains for every $i,j \in \{1,\dots,k\}$ and $f,g,h, \in \fsig\rho$
    either the literal $f(x_i) = h(g(x_j))$ or its negation.
\item $\Delta^=(y\bx) \in \funformula{1}{\sigma}$ is a nonempty
    conjunction of positive literals of the
    form $f(y)=g(x_i)$ with $f,g \in \fsig\sigma$ and $i \in \{1,\dots,k\}$,
\item $\Delta^{\neq}(y\bx) \in \funformula{1}{\sigma}$ 
    is a conjunction of negative literals of the
    form~$f(y)\neq g(x_i)$ with $f,g \in \fsig\sigma$ and $i \in \{1,\dots,k\}$.
\end{enumerate}
We denote the set of all such canonical conjunctive clauses
by $\C(k,\sigma,\rho)$.
\end{definition}

We call the literals in $\Delta^=(y\bx)$ and $\Delta^{\neq}(y\bx)$
the \emph{mixed} literals of a canonical conjunctive clause.
The requirement that $\Delta^=(y\bx)$ is nonempty
is a technical assumption that we will need later.
In the following lemma the literal $\fapx(y) = \fapx(x_1)$
is needed to make sure this assumption is fulfilled.
Over the course of this section we will gradually decompose a quantifier-free
formula into more and more simple combinations of canonical conjunctive clauses.

\begin{lemma}\label{lem:canonical-formula}
For two signatures $\sigma$, $\rho$ with $\sigma \subseteq \rho$ and
a given quantifier-free formula $\phi(y\bx) \land \fapx(y) = \fapx(x_1) \in \funformula{1}{\sigma}$
one can compute a set of canonical conjunctive clauses $\Omega \subseteq \C(|\bx|,\sigma,\rho)$
such that 
for every $\vG \in \fungraph{\rho}$ and every tuple
of vertices $v\bu \in V(\vG)^{|y\bx|}$ 
$$
\vG\models\phi(v\bu)\land \fapx(v) = \fapx(u_1) \text{~~iff~~}
\vG\models\omega(v\bu)\text{ for some }\omega\in\Omega.
$$
Furthermore, $\Omega$ is mutually exclusive in the sense that
for every $\vG \in \fungraph{\rho}$ and tuple $v\bu \in V(\vG)^{|y\bx|}$ 
there is at
most one $\omega\in\Omega$ with $\vG\models\omega(v\bu)$.
\end{lemma}
\begin{proof}
    We can assume $\phi(y\bx)$ to be given in disjunctive normal form.
    Consider a conjunctive clause $\omega(y\bx)$ of this normal form and any literal $l(y\bx)$.
    We can replace $\omega(y\bx)$ with two clauses
    $\omega(y\bx) \land l(y\bx)$ and $\omega(y\bx) \land \neg l(y\bx)$.
    The result is still a disjunctive normal form of $\phi(y\bx)$.
    We can therefore assume that every clause of $\phi(y\bx)$ contains
    \begin{itemize}
        \item for every valid literal $l$ in $\funformula{1}{\sigma}$ 
            with free variables from $y\bx$ either $l$ or $\neg l$,
        \item for every literal $l$ of the form 
            $h(f(y)) = g(y)$ or $h(f(x_i)) = g(x_j)$
            with $f,g,h \in \fsig\rho$ and $i,j \in \{1,\dots,|\bx|\}$
            either $l$ or $\neg l$.
    \end{itemize}
    Let $\Omega$ be the set of conjunctive clauses of $\phi(y\bx)$.
    Any two clauses in $\Omega$ disagree in at least one literal.
    Thus, they cannot be satisfied by the same interpretation.
    This means $\Omega$ is mutually exclusive.
    Furthermore, since $\Omega$ is obtained from a disjunctive normal form,
    for every $\vG \in \fungraph{\rho}$ 
    and every $v\bu \in V(\vG)^{|y\bx|}$
    $$
    \vG\models\phi(v\bu)\text{~~iff~~}
    \vG\models\omega(v\bu)\text{ for some }\omega\in\Omega.
    $$

    However, the formulas in $\Omega$ are not yet canonical conjunctive clauses.
    We fix a clause from $\Omega$ and decompose it into four subclauses
    $\tau(y)\land\psi(\bx)\land \Delta^=(y\bx)\land\Delta^{\neq}(y\bx)$,
    where $\tau(y)$, $\psi(\bx)$ contain all literals depending
    on $y$ and $\bx$,
    and $\Delta^=(y\bx)$, $\Delta^{\neq}(y\bx)$
    contain the remaining positive and negative literals, respectively.

    The clauses $\tau(y)$, $\psi(\bx)$ are of the form mentioned in \Cref{def:canonical},
    while $\Delta^=(y\bx)$, $\Delta^{\neq}(y\bx)$ might not.
    We will modify them to fit \Cref{def:canonical}.
    Besides the allowed literals,
    $\Delta^=(y\bx)$ may also contain literals of the form
    $f(x_i) = y$, $f(y) = x_i$ or $x_i=y$.
    Using the identify function $f_{\text{id}}$,
    we can artificially turn them into equivalent literals
    $f(x_i) = f_{\text{id}}(y)$, $f(y) = f_{\text{id}}(x_i)$, or
    $f_{\text{id}}(x_i) = f_{\text{id}}(y)$ of the allowed form.
    We proceed similarly for $\Delta^{\neq}(y\bx)$.
    We also add the literal $\fapx(y) = \fapx(x_1)$ to $\Delta^=(y\bx)$.
    Therefore, $\Delta^=(y\bx)$ is nonempty.
    Now all clauses are of the form stated in \Cref{def:canonical}.
    We apply this procedure to every clause in $\Omega$.
    Then $\Omega \subseteq \C(|\bx|,\sigma,\rho)$.
    Because we added $\fapx(y) = \fapx(x_1)$
    to every canonical conjunctive clause, we have
    $$
    \vG\models\phi(v\bu) \land \fapx(v) = \fapx(u_1) \text{ iff }
    \vG\models\omega(v\bu)\text{ for some }\omega\in\Omega.
    $$
\end{proof}

In the previous lemma, it would have been okay to discard unsatisfiable
formulas from~$\Omega$.
When we go from a functional structure $\vG$ to its 1-transitive fraternal augmentation, 
new function symbols are inserted, representing transitive and fraternal arcs.
We will later argue that it is okay also to remove those formulas
which are not satisfied by any 1-transitive fraternal augmentation
or extension thereof.
Since all 1-transitive fraternal augmentations have a certain structure,
we can discard more formulas.
The next definition formally captures these concepts. 


\begin{definition}\rm
    Let $\sigma$, $\rho$ be signatures with $\sigma \subseteq \rho$ and $\vG' \in \fungraph{\rho}$.
    We say $\vG'$ is a \emph{$\sigma$-$\rho$-expansion} if there exists 
    $\vG \in \fungraph{\sigma}$
    such that $\vG'$ is an expansion of the $1$-transitive fraternal augmentation of $\vG$.
    We also say $\vG'$ is a $\sigma$-$\rho$-expansion of $\vG$.
    A canonical conjunctive clause 
    $\omega(y\bx) \in \C(|\bx|,\sigma,\rho)$ is \emph{$\sigma$-$\rho$-unsatisfiable}
    if $\vG' \not\models \omega(v\bu)$ holds
    for every $\sigma$-$\rho$-expansion $\vG' \in \fungraph{\rho}$
    and every $v\bu \in V(\vG')^{|y\bx|}$.
\end{definition}

In the next step we further simplify the formulas by reducing the
number of mixed positive literals from an arbitrary number down to one.

\begin{lemma}\label{lem:step1}
    Let $\sigma$, $\rho$ be signatures with $\sigma \subseteq \rho$ and
    $\omega(y\bx) = \tau(y)\land\psi(\bx)\land \Delta^=(y\bx)\land\Delta^{\neq}(y\bx) \in \C(|\bx|,\sigma,\rho)$.
    There exists an algorithm that either computes
    a literal $f(y)=g(x_i)\in\Delta^=(y\bx)$ such that
    $$
    \omega(y\bx)
    \equiv
    \tau(y)\land\psi(\bx)\land f(y)=g(x_i)\land\Delta^{\neq}(y\bx)
    $$
    or concludes that $\omega(y\bx)$ is $\sigma$-$\rho$-unsatisfiable.
\end{lemma}

\begin{proof}
By definition, $\Delta^{=}(y\bx)$ is nonempty.
If it contains only one literal, we do not need to do anything.
Let us assume there are two literals $f(y)=g(x_i)$ and $f'(y)=g'(x_j)$ in
$\Delta^=(y\bx)$.

Let $\vG'$ be a $\sigma$-$\rho$-expansion and $v \in V(\vG)$.
Since $f,f' \in \fsig\sigma$ and by the fraternal rule,
there exists a function $h \in \fsig\rho$ such that either 
$h^{\vG}(f^{\vG}(v)) = f'^{\vG}(v)$ or $f^{\vG}(v) = h^{\vG}(f'^{\vG}(v))$.
Thus, if $\omega(y\bx)$ is $\sigma$-$\rho$-satisfiable
then $\tau(y)$ either contains
$h(f(y)) = f'(y)$ or $f(y) = h(f'(y))$
for some function $h \in \fsig\rho$.
Let us assume it is $h(f(y))=f'(y)$ because the other case is similar.  
Then there must be $h(g(x_i))=g'(x_j)$ present in $\tau(y)$
(or the formula is unsatisfiable and in particular
$\sigma$-$\rho$-unsatisfiable).  From $f'(y)=h(f(y))$,
$h(g(x_i))=g'(x_j)$, and $f(y)=g(x_i)$ follows $f'(y)=g'(x_j)$.
Therefore we can remove $f'(y)=g'(x_j)$ from $\Delta^=(y\bx)$.
We repeat this procedure as long as $\Delta^=(y\bx)$ contains at
least two literals.
\end{proof}

The next simplification gets rid of some of the negative mixed
literals, however, not all of them.  The remaining ones have a special
relation to the rest of the conjunctive clause that will cause them
later on to have only a small influence on the counting term.

\begin{lemma}\label{lem:step2}
    Let $\sigma,\rho$ be signatures with $\sigma \subseteq \rho$ and
    $\omega(y\bx) = \tau(y)\land\psi(\bx)\land f(y)=g(x_i) \land\Delta^{\neq}(y\bx) \in \C(|\bx|,\sigma,\rho)$.
    We define $\Delta^{\neq}_\tau(y\bx)$ to be the set of all literals of the form $f'(y)\neq g'(x_j)$
    in $\Delta^{\neq}(y\bx)$
    such that $\tau(y)$ contains
    $h(f(y)) \neq f'(y)$ for all $h \in \fsig\rho$.
    There exists an algorithm that concludes either that
    $$
    \omega(y\bx) \equiv
    \tau(y)\land\psi(\bx)\land f(y)=g(x_i)\land\Delta_\tau^{\neq}(y\bx)
    $$
    or that $\omega(y\bx)$ is $\sigma$-$\rho$-unsatisfiable.
\end{lemma}

\begin{proof}
Let $f'(y)\neq g'(x_j)$ be a literal
that is contained in $\Delta^{\neq}(y\bx)$,
but not in $\Delta_\tau^{\neq}(y\bx)$.
We argue that we can either safely remove it or the formula is unsatisfiable.
By our assumption 
there exists $h \in \fsig\rho$ such that $h(f(y))\neq f'(y)\not\in\tau(y)$.
Since $\tau(y)$ is complete, $h(f(y))=f'(y)\in\tau(y)$.
Since $\psi(\bx)$ is also complete,
it either contains the literal $h(g(x_i)) = g'(x_j)$
or its negation $h(g(x_i))\neq g'(x_j)$.

Assume $h(g(x_i))\neq g'(x_j)\in\psi(\bx)$.
The literal $f(y)=g(x_i)$ with $h(f(y))=f'(y)$ and $h(g(x_i))\neq
g'(x_j)$ implies $f'(y)\neq g'(x_j)$, so we can safely remove it from
$\Delta^{\neq}(y\bx)$.

Assume $h(g(x_i))=g'(x_j)\in\psi(\bx)$.
The literal $f(y)=g(x_i)$ with $h(f(y))=f'(y)$ implies
$h(g(x_i)) = f'(y)$.
On the other hand, $f'(y) \neq g'(x_j)$ with $h(g(x_i))=g'(x_j)$ 
implies $h(g(x_i)) \neq f'(y)$.
Henceforth, $\omega(y\bx)$ is unsatisfiable.
\end{proof}

The following lemma aggregates the results from 
Lemmas~\ref{lem:canonical-formula}, \ref{lem:step1}, and~\ref{lem:step2}.

\begin{lemma}\label{lem:step12}
For two signatures $\sigma$, $\rho$ with $\sigma \subseteq \rho$ and
a given quantifier-free formula $\phi(y\bx) \land \fapx(y) = \fapx(x_1) \in \funformula{1}{\sigma}$
one can compute a set of canonical conjunctive clauses $\Omega\subseteq \C(|\bx|,\sigma,\rho)$
with the following properties:
\begin{enumerate}
    \item
    Every formula $\omega \in \Omega$ has the form
    $\tau(y)\land\psi(\bx)\land f(y)=g(x_i)\land\Delta_\tau^{\neq}(y\bx)$.
    We require
    for every literal of the form $f'(y)\neq g'(x_j)$ in $\Delta^{\neq}_\tau(y\bx)$
    and every $h \in \fsig\rho$ that $h(f(y)) \neq f'(y) \in \tau(y)$.

    \item
    For every $\sigma$-$\rho$-expansion $\vG$ 
    and every tuple $v\bu \in V(\vG)^{|y\bx|}$ holds
    $$
    \vG\models\phi(v\bu) \land \fapx(v) = \fapx(u_1) \text{~~iff~~}
    \vG\models\omega(v\bu)\text{ for some }\omega\in\Omega.
    $$

    \item
    $\Omega$ is mutually exclusive in the sense that
    for every $\sigma$-$\rho$-expansion $\vG$ 
    and every tuple $v\bu \in V(\vG)^{|y\bx|}$ 
    there is at most one $\omega\in\Omega$ with $\vG\models\omega(v\bu)$.
\end{enumerate}
\end{lemma}
\begin{proof}
Let $\Omega \subseteq \C(|\bx|,\sigma,\rho)$ be the set computed by \Cref{lem:canonical-formula}.
This set already satisfies properties 2 and~3.
We will modify it such that it also satisfies the first property.
Let $\omega(y\bx) \in \Omega$.
We first apply the algorithm from \Cref{lem:step1} and then 
(assuming it was not concluded that $\omega(y\bx)$ is
$\sigma$-$\rho$-unsatis\-fi\-able)
the one from \ref{lem:step2}.
This either yields a formula
$$
\omega(y\bx) \equiv \tau(y)\land\psi(\bx)\land f(y)=g(x_i)\land\Delta_\tau^{\neq}(y\bx)
$$
that meets the requirements of the first property or concludes that
$\omega(y\bx)$ is $\sigma$-$\rho$-unsatis\-fi\-able.
If the formula is $\sigma$-$\rho$-unsatisfiable, we can remove
if from $\Omega$ and properties 2 and~3 remain true.
We apply this procedure for every formula in $\Omega$.
Then property~1 is also satisfied.
\end{proof}

The following lemma shows that we can evaluate a simple counting
term of the form
$\#y\, \tau(y) \land f(y) = u \land f'(y) = u'$
for all values of $u$ and $u'$ in linear time.
Since there are a quadratic number of tuples $u$, $u'$, we only write
down those tuples where the counting term is non-zero.
Later we will need these numbers for an algorithm that can identify
places where we need to add arcs to the functional structure in order
to be able to eliminate certain negative mixed literals that have a
non-negligible contribution to the value of a counting term.
We will need the same lemma also in the proof that
\emph{exact} counting is possible for formulas of a certain shape.

\begin{lemma}\label{lem:actualcounting}
Let ${\cal C}\subseteq \fungraph{\sigma}$ be a class with bounded expansion,
$\tau(y) \in \funformula{2}{\sigma}$ and $f,f' \in \fsig\sigma$.
For an input $\vG \in \cal C$ the list of triples
$$
\{\,(u,u',c) \mid u,u' \in V(\vG),
   c = [[\#y\, \tau(y) \land f(y) = u \land f'(y) = u']]^{\vG}, c>0\,\}
$$
can be computed in time $O(\norm{\vG})$.
\end{lemma}

\begin{proof}
We define a ``counter'' $c(u,u')$ for $u,u'\in V(\vG)$ as
$$
c(u,u') = {\bigm|}\{\,v \in V(\vG)
    \mid \vG \models \tau(v)\land f(v)=u\land f'(v)=u'\,\}{\bigm|}.
$$
The following algorithm
obviously computes $c(u,u')$.
\medskip
\begin{algorithmic}
\For{$v\in V(\vG)$ with $\vG\models\tau(v)$}
\State $c(f(v),f'(v)) \gets c(f(v),f'(v))+1$
\EndFor
\end{algorithmic}

\medbreak
Computing $c(u,u')$ for all $u,u'\in V(\vG)$ would be easy in quadratic
time, but is also possible in linear time.  
Kazana and Segoufin showed that we
can enumerate all vertices $v$ with $G\models\tau(v)$ in time linear
in $\norm{\vG}$~\cite{KazanaS2013}.
The only issue with this short piece of code is how to store the
counters $c(u,u')$.  There is a quadratic number of them, although
most of them are left to be zero and we are only interested in those
with a positive count.  One possibility
is delaying the increment of the counters to the end.
Instead of carrying out the commands $c(f(v),f'(v)) \gets
c(f(v),f'(v))+1$ immediately we store them in an array of linear
length.  At the end we can sort this array in linear time (e.g., by a
combination of radix- and bucket-sort) and then combine blocks of
identical commands while counting their sizes.  What remains is a list of
the positive counters together with their respective values.
\end{proof}

The next step is to actually compute an expansion that is prepared in
such a way that every mixed negative literal from every possible
relevant conjunctive clause that influences the underlying counting
term significantly can be removed.  While Lemma~\ref{lem:aug} does the
preparation by adding flip arcs,
Lemma~\ref{lem:aug2} shows that the resulting expansion has the
desired property.

\begin{lemma}\label{lem:aug}
Let $\cal C \subseteq \fungraph{\sigma}$ be a class with bounded
expansion and $\varepsilon>0$.  There exists a signature
$\rho\supseteq\sigma$ and a class $\cal C' \subseteq \fungraph{\rho}$
with bounded expansion such that for every $\vG \in \cal C$ one can
compute a $\sigma$-$\rho$-expansion $\vG' \in \cal C'$ of $\vG$ with
the following property in time $O(\norm{\vG})$:

If there exist $u,u' \in V(\vG')$, a
quantifier-free formula~$\tau(y) \in \funformula{2}{\sigma}$,
and $f$, $f'\in \fsig\sigma$ such that
$[[\#y\, \tau(y)\land f(y)=u\land f'(y)=u']]^{\vG'} >
\epsilon[[\#y\,\tau(y)\land f(y)=u]]^{\vG'}$
then there exists $h \in \fsig\rho$ with $h^{\vG'}(u) = u'$ 
(i.e., $\vG'$ contains an arc~$u'u$).
\end{lemma}

\begin{proof}
We have to show that we can identify all pairs $u',u$ that need a new arc
in linear time and that the resulting functional representation
belongs to a class with bounded expansion.
We start with the second task.  

We will express $\vG'$ by first doing a $1$-transitive fraternal
augmentation on $\vG$ and then adding all
necessary remaining arcs using a bounded number of flip augmentations.
By \Cref{cor:funcBE}, $\vG'$ then belongs to a class with bounded expansion.

Let $u, u' \in V(\vG)$ such that there needs to be an arc $u'u$ in $\vG'$.
This means in particular that $[[\#y\, \tau(y)\land f(y)=u\land f'(y)=u']]^{\vG'} > 0$,
i.e., there is a vertex $v \in V(\vG)$ such that the
arcs $u'v$ and $uv$ are present in $\vG$.
By the fraternal rule, either the arc $u'u$ or $uu'$ exists
in the $1$-transitive fraternal augmentation of $\vG$.

If the arc $u'u$ is already present in the $1$-transitive fraternal
augmentation, we do not need to do anything.
But if only the arc $uu'$ is there, we need to flip it.
Since ${\cal C}$ is a class of
bounded expansion, there must be a constant bound on the indegree of
all graphs in~$\cal C$.  Let us say this bound is~$d$.
We only stay within a class with bounded expansion if
the indegree remains bounded after flipping the arcs.
We fix a vertex $u \in V(\vG)$ and show that only a constant number of edges
need to be oriented towards $u$.
The set $U' = \{\,u' \in V(\vG) \mid [[\#y\, \tau(y)\land f(y)=u\land
f'(y)=u']]^\vG >\epsilon[[\#y\,\tau(y)\land f(y)=u]]^\vG\,\}$
contains all vertices $u'$ for which an arc $u'u$ needs to be present in $\vG'$.
Let $A=\{\,v \in V(\vG) \mid \vG \models \tau(v)\land f(v)=u\,\}$.  
For every $u' \in U'$ there are at least $\varepsilon |A|$ arcs
going from $u'$ to $A$ in $\vG$.
Therefore, there are at least
$\epsilon|A|\cdot|U'|$ many arcs going into the set~$A$ and there must
be one vertex in $A$ that receives at least $\epsilon|U'|$ of them.
As this number must be smaller than $d$, we can conclude that $|U'|\leq
d/\epsilon$.
Hence, $\vG'$ can be obtained using a $1$-transitive fraternal
augmentation and at most $d/\epsilon$ many flip operations.

It remains to be shown that this construction can be carried out in
linear time.  There is only a constant number of combinations of
$f,f',\tau$ because $|\sigma|$ and the number of non-equivalent
quantifier-free formulas in $\funformula2\sigma$ are constant.
For each such combination we can compute the lists
\begin{align*}
l_1 & =\{\,(u,u',c) \mid u,u' \in V(\vG),
   c = [[\#y\, \tau(y) \land f(y) = u \land f'(y) = u']]^{\vG}, c>0\,\}, \\
l_2 & =\{\,(u,c) \mid u \in V(\vG),
c = [[\#y\, \tau(y) \land f(y) = u]]^{\vG}, c>0\,\}
\end{align*}
in linear time by Lemma~\ref{lem:actualcounting}.
Using $l_1$ and $l_2$ it is easy to determine all $u',u$ with
$[[\#y\, \tau(y)\land f(y)=u\land f'(y)=u']]^\vG
>\epsilon[[\#y\,\tau(y)\land f(y)=u]]^\vG$.
\end{proof}

\begin{lemma}\label{lem:aug2}
Let $\cal C \subseteq \fungraph{\sigma}$ be a class with bounded
expansion and $\varepsilon>0$.  There exists a signature
$\rho\supseteq\sigma$ and a class $\cal C' \subseteq \fungraph{\rho}$
with bounded expansion such that for every $\vG \in \cal C$ one can
compute a $\sigma$-$\rho$-expansion $\vG' \in \cal C'$ of $\vG$ with
the following property in time $O(\norm{\vG})$:

For every formula $\tau(y)\land\psi(\bx)\land f(y)=g(x_i)\land f'(y)=g'(x_j) \in \C(|\bx|,\sigma,\rho)$
with $h(f(y)) \neq f'(y) \in \tau(y)$ for all $h \in \fsig\rho$
and for every $\bu \in V(\vG)^{|\bx|}$
holds
$$
[[\#y\, \tau(y)\land \psi(\bu) \land f(y)=g(u_i)\land f'(y)= g'(u_j)]]^{\vG'} \le
\epsilon[[\#y\,\tau(y)\land \psi(\bu) \land f(y)=g(u_i)]]^{\vG'}.
$$
\end{lemma}
\begin{proof}
    For a given input $\vG$, we construct $\vG' \in \fungraph{\rho}$ according to \Cref{lem:aug}.
    Consider a formula $\tau(y)\land\psi(\bx)\land f(y)=g(x_i)\land f'(y)=g'(x_j)$
    and a tuple $\bu \in V(\vG)^{|\bx|}$ as specified above.
    For convenience we define $u = g(u_i)$ and $u' = g'(u_j)$.
    We assume there exists $v \in V(\vG')$ such that
    $\vG' \models \tau(v)\land \psi(\bu) \land f(v)=u \land f'(v)= u'$,
    since otherwise the left-hand side of the equation is zero and
    the statement is trivially true.
    For every $h \in \fsig\rho$ we have
    $h(f(y)) \neq f'(y) \in \tau(y)$.
    This implies
    $h^{\vG'}(u) = h^{\vG'}(f^{\vG'}(v)) \neq f'^{\vG'}(v) = u'$
    for every $h \in \fsig\rho$.
    According to \Cref{lem:aug},
    $$
    [[\#y\, \tau(y)\land f(y)=u\land f'(y)= u']]^{\vG'} \le
    \epsilon[[\#y\,\tau(y) \land f(y)=u]]^{\vG'},
    $$
    since otherwise $h^{\vG'}(u) = u'$ for some $h \in \fsig\rho$.
    The result follows since $\vG' \models \psi(\bu)$.
\end{proof}

The following lemma presents a way to reduce the functional depth of a formula.
It is taken from~\cite{KazanaS2013} with the notation changed to ours.
While~\cite{KazanaS2013} does not mention that the formula with reduced depth
can be computed, it clearly follows from their construction.

\begin{proposition}[Lemma 4, \cite{KazanaS2013}]\label{prop:funcdepth}
    For every 
    quantifier-free formula $\phi(\bx)$ with signature $\sigma$
    and class $\cal C \in \fungraph{\sigma}$ with bounded expansion
    there exists a signature $\rho \supseteq \sigma$,
    a class $\cal C' \subseteq \fungraph{\rho}$ with bounded expansion
    and a quantifier-free formula $\phi'(\bx) \in \funformula{1}{\rho}$
    such that the following properties are satisfied:
    \begin{itemize}
    \item
        The formula $\phi'$ can be computed from $\phi$.
    \item
        For every $\vG \in \cal C$ we can compute
        in time $O(\norm{\vG})$
        an expansion $\vG' \in \cal C'$ of $\vG$ 
        such that for every $\bu \in V(\vG)^{|\bx|}$
        holds $\vG \models \phi(\bu)$ iff $\vG' \models \phi'(\bu)$.
    \end{itemize}
\end{proposition}

While ignoring a single negative mixed literal is now guaranteed to 
change the value of a counting term only by a little, the next Lemma
helps to estimate the influence of removing all negative mixed
literals at once.

\begin{lemma}\label{lem:technicalepsilon}
    Let $\omega(y), l_1(y),\dots,l_t(y)$ be formulas,
    $\vG$ be a functional representation,
    and $\varepsilon > 0$.
    If 
    for every $i \in \{1,\dots,t\}$
    with $\varepsilon' = \min(1,\varepsilon)/2t$
    $$
    [[\#y\, \omega(y) \land l_i(y) ]]^{\vG} \le
    \epsilon' [[\#y\,\omega(y)]]^{\vG}
    $$
    then
    $$
    [[\#y\,\omega(y)\land \bigwedge_{i=1}^t \neg l_i(y)]]^\vG
     \le
    [[\#y\,\omega(y)]]^\vG
    \le (1+\varepsilon)
    [[\#y\,\omega(y)\land \bigwedge_{i=1}^t \neg l_i(y)]]^\vG.
    $$
\end{lemma}
\begin{proof}
    Every assignment $v \in V(\vG)$ for $y$ satisfies
    either $\bigwedge_{i=1}^t \neg l_i(v)$
    or $l_i(v)$ for at least one $i \in \{1,\dots,t\}$.
    Therefore
    $$
    [[\#y\,\omega(y)]]^{\vG} \le
    [[\#y\,\omega(y) \land \bigwedge_{i=1}^t \neg l_i(y)]]^{\vG} + 
    \sum_{i = 1}^t
    [[\#y\,\omega(y) \land l_i(y)]]^{\vG}.
    $$
    Using our initial assumption, this means
    $$
    [[\#y\,\omega(y)]]^{\vG} \le
    [[\#y\,\omega(y) \land \bigwedge_{i=1}^t \neg l_i(y)]]^{\vG} + 
    t\epsilon' [[\#y\,\omega(y)]]^{\vG},
    $$
    and thus
    $$
    (1-t\varepsilon')
    [[\#y\,\omega(y)]]^\vG
    \le 
    [[\#y\,\omega(y)\land \bigwedge_{i=1}^t \neg l_i(y)]]^\vG
    \le [[\#y\,\omega(y)]]^\vG.
    $$
    Let $a,b > 0$ with $(1 - t\varepsilon') a \le b \le a$.
    Dividing by $a$ yields
    $1- t\varepsilon' \le b/a \le 1$,
    and then taking the reciprocal gives us
    $1 \le a/b \le 1/(1- t\varepsilon') = 1/(1- \min(1,\varepsilon)/2) \le 1 + \varepsilon$.
    We multiply with $b$ and obtain
    $b \le a \le (1 + \varepsilon)b$.
    This yields the statement of this lemma.
\end{proof}

We finally arrive at the point where we can approximate a counting
term with a sum of simple terms of the form
$\#y\,\tau(y)\land\psi(\bx)\land f(y)=g(x_i)$.  While they still have
many free variables, $\psi(\bx)$ does not depend on $y$ and can therefore be pulled outside the counting quantifier.
This leaves a counting term with a single free variable
that we can evaluate using \Cref{lem:actualcounting} in linear time.

\begin{lemma}\label{lem:apx}
    Let $\cal C \subseteq \fungraph{\sigma}$ be a class with bounded expansion
    and $\varepsilon > 0$.
    One can compute for every quantifier-free formula $\phi(y\bx)$
    with signature $\sigma$ 
    a signature $\rho\supseteq\sigma$ and
    a set of conjunctive clauses $\Omega \subseteq \funformula{2}{\rho}$
    of the form $\tau(y)\land \psi(\bx)\land f(y)=g(x_i)$ 
    with the following property:

    There exists a class $\cal C' \subseteq \fungraph{\rho}$ with bounded expansion
    such that for every $\vG \in \cal C$ one can compute in time $O(\norm{\vG})$
    an expansion $\vG' \in \cal C'$ of $\vG$ 
    such that for every $\bu \in V(\vG)^{|\bx|}$
    $$
    [[\#y\,\phi(y\bu)]]^\vG
    \le \sum_{\omega\in\Omega}[[\#y\,\omega(y\bu)]]^{\vG'}
    \le (1+\varepsilon)
    [[\#y\,\phi(y\bu)]]^\vG.
    $$
\end{lemma}
\begin{proof}
    Assume we are given a formula $\phi(y\bx)$ and a functional representation
    $\vG \in \cal C$.
    We pick a vertex $w \in V(\vG)$ and add a new function symbol $\fapx$ to our signature $\sigma$ and input structure
    with $\fapx^{\vG}(v) = w$ for every $v \in V(\vG)$.
    For the underlying directed graph of $\vG$, this amounts to making $w$ an apex vertex.
    Thus by \Cref{def:be-functional}, $\vG$ still comes from a class with bounded expansion.

    Using \Cref{prop:funcdepth} we can replace $\phi(y\bx)$
    with another quantifier-free formula with functional depth one.
    The price we have to pay is replacing $\vG$
    with an expanded functional representation, which is still from a
    class with bounded expansion.
    Both the new formula and the new functional representation can be computed in the desired time.
    Therefore, from now on, we can assume without loss of generality
    that $\phi(y\bx)$ is a quantifier-free formula with signature $\sigma$ and functional depth one
    and that there exists a function symbol $\fapx \in \fsig\sigma$ with
    $\fapx^{\vG}(v) = \fapx^{\vG}(v')$ for every $v,v' \in V(\vG)$.

    \Cref{lem:aug2} with $\varepsilon' = \min(1,\varepsilon)/2|\sigma|^2|\bx|$
    gives us a signature $\rho$ and a bounded expansion class $\cal C' \subseteq \fungraph{\rho}$.
    We compute in time $O(\norm{\vG})$ the $\sigma$-$\rho$-expansion $\vG' \in \cal C'$ of $\vG$ with properties as in \Cref{lem:aug2}.
    Let $\Omega' \subseteq \C(|\bx|,\sigma,\rho)$ be the set of canonical conjunctive clauses obtained from $\phi(y\bx) \land \fapx(y) = \fapx(x_1) $ with the algorithm from \Cref{lem:step12}.
    Property 2 and 3 of \Cref{lem:step12} together imply for every $\bu \in V(\vG)^{|\bx|}$
    \begin{equation}\label{eq:apx0}
    [[ \cnt{y} \phi(y\bu) ]]^{\vG}
    = [[ \cnt{y} \phi(y\bu) \land \fapx^{\vG}(y) = \fapx^{\vG}(u_1)]]^{\vG}
    = \sum_{\omega' \in \Omega'} [[ \cnt{y} \omega'(y\bu) ]]^{\vG'}.
    \end{equation}

    We consider a clause 
    $\omega'(y\bx) = \tau(y)\land\psi(\bx)\land f(y)=g(x_i)\land \Delta_\tau^{\neq}(y\bx) \in \Omega'$,
    where $\Delta_\tau^{\neq}(y\bx)$ is of the form $\bigwedge_{i=1}^t \neg l_i(y\bx)$.
    \Cref{lem:aug2} states (using the first property of \Cref{lem:step12})
    for $1 \le i \le t$ and $\bu \in V(\vG)^{|\bx|}$
    $$
    [[\#y\, \tau(y)\land \psi(\bu) \land f(y)=g(u_i)\land l_i(y\bu)]]^{\vG'} \le
    \epsilon' [[\#y\,\tau(y)\land \psi(\bu) \land f(y)=g(u_i)]]^{\vG'}.
    $$
    Since the literals of $\Delta_\tau^{\neq}(y\bx)$ have
    the form $f(y) \neq g(x_i)$
    with $f,g \in \fsig\sigma$, we have $t \le |\sigma|^2|\bx|$.
    Hence, by \Cref{lem:technicalepsilon},
    \begin{equation}\label{eq:apx1}
    [[\#y\,\omega'(y\bu)]]^{\vG'}
    \le 
    [[\#y\,\tau(y)\land\psi(\bu)\land f(y)=g(u_i)]]^{\vG'}
    \\
    \le
    (1 + \varepsilon) [[\#y\,\omega'(y\bu)]]^{\vG'}.
    \end{equation}
    Combining (\ref{eq:apx0}) and (\ref{eq:apx1}) yields our result.
\end{proof}

The last step to reach this section's goal is to actually evaluate the
simple counting terms in the sum of the last lemma and store
the results as ``weights'' at the individual vertices.
Finally, this allows us to approximate a counting term $\#y\,\phi(y\bu)$
using combinations of quantifier-free first-order formulas
and the calculated weights.
We consider the following theorem the main technical contribution
of this paper,
as the whole upcoming quantifier elimination procedure builds upon it.
Since negative summands lead to the problem of cancellation and therefore bad approximations,
a lot of effort has been spent in this section to make sure
that none of the summands $c_{\omega,i}(u_i)$ are negative.
A similar result with negative summands is easier to prove 
and can be found in \Cref{thm:mainExact}.

\begin{theorem}\label{thm:main}
    Let $\cal C \subseteq \fungraph{\sigma}$ be a class with bounded expansion
    and $\varepsilon > 0$.
    One can compute for every 
    quantifier-free formula
    $\phi(y\bx)$
    with signature $\sigma$
    a set of conjunctive clauses
    $\Omega$ 
    with free variables $\bx$ and signature
    $\rho\supseteq\sigma$
    that satisfies the following property:

    There exists a class $\cal C' \subseteq \fungraph{\rho}$ with bounded expansion
    such that for every $\vG \in \cal C$ one can compute in time $O(\norm{\vG})$
    an expansion $\vG' \in \cal C'$ of $\vG$ 
    and functions $c_{\omega,i}(v) \colon V(\vG) \to \N$ for $\omega \in \Omega$ and $i \in \{1,\dots,|\bx|\}$
    such that for every $\bu \in V(\vG)^{|\bx|}$
    there exists exactly one formula $\omega \in \Omega$
    with $\vG' \models \omega(\bu)$. For this formula
    $$
    [[\#y\,\phi(y\bu)]]^\vG
    \le \sum_{i=1}^{|\bx|}c_{\omega,i}(u_i)
    \le (1+\varepsilon)
    [[\#y\,\phi(y\bu)]]^\vG.
    $$
\end{theorem}
\begin{proof}
    We use \Cref{lem:apx} to construct $\Omega' \subseteq \funformula{2}{\rho}$ and
    $\vG' \in \cal C' \subseteq \fungraph{\rho}$ such that
    for every $\bu \in V(\vG)^{|\bx|}$
    \begin{equation}\label{eq:x}
    [[\#y\,\phi(y\bu)]]^\vG
    \le \sum_{\omega'\in\Omega'}[[\#y\,\omega'(y\bu)]]^{\vG'}
    \le (1+\varepsilon)
    [[\#y\,\phi(y\bu)]]^\vG.
    \end{equation}
    Let $\Omega \subseteq \funformula{2}{\rho}$ 
    be the set of all complete conjunctive clauses with functional depth two, signature $\rho$ and free variables $\bx$.
    This set has two important properties:
    First, for every $\bu \in V(G)^{|\bx|}$ there exists exactly one 
    $\omega \in \Omega$ with $\vG' \models \omega(\bu)$.
    Secondly, for every $\omega \in \Omega$ and conjunctive clause $\psi(\bx) \in \funformula{2}{\rho}$
    either $\omega \models \psi$ or $\omega \models \neg\psi$.

    Let now $\bu \in V(\vG)^{|\bx|}$,
    $\tau(y) \land \psi(\bx) \land f(y) = g(x_i) \in \Omega'$,
    and $\omega \in \Omega$ such that $\vG' \models \omega(\bu)$.
    If $\omega \models \neg\psi$ then
    $[[\#y\, \tau(y) \land \psi(\bu) \land f(y) = g(u_i)]]^{\vG'} = 0$.
    If $\omega \models \psi$ then
    $[[\#y\, \tau(y) \land \psi(\bu) \land f(y) = g(u_i)]]^{\vG'} =
    [[\#y\, \tau(y) \land f(y) = g(u_i)]]^{\vG'}$.
    Using this observation, we define for every $\omega \in \Omega$ and $i \in \{1,\dots,|\bx|\}$ a set $\Gamma_{\omega,i}$
    by iterating over all formulas $\omega \in \Omega$ and 
    $\tau(y) \land \psi(\bx) \land f(y) = g(x_i) \in \Omega'$
    and adding $\tau(y) \land f(y) = g(x_i)$ to $\Gamma_{\omega,i}$
    if $\omega \models \psi$.
    Now for every $\bu \in V(\vG)^{|\bx|}$
    there exists exactly one formula $\omega \in \Omega$
    with $\vG' \models \omega(\bu)$, and for this formula\looseness-1
    \begin{equation}\label{eq:y}
    \sum_{\omega'\in\Omega'}[[\#y\,\omega'(y\bu)]]^{\vG'} =
    \sum_{i = 1}^{|\bx|}
    \sum_{\tau(y) \land f(y) = g(x_i) \in\Gamma_{\omega,i}}
    [[\#y\,\tau(y) \land f(y) = g(u_i)]]^{\vG'}.
    \end{equation}

    Let us fix one set $\Gamma_{\omega,i}$.
    For every formula $\tau(y) \land f(y) = g(x_i) \in \Gamma_{\omega,i}$
    we use \Cref{lem:actualcounting} to construct a function $c$ with
    $c(v) = [[\#y\,\tau(y) \land f(y) = g(v)]]^{\vG'}$.
    Let $c_{\omega,i}$ be the sum over all such functions
    for formulas in $\Gamma_{\omega,i}$.
    Then
    \begin{equation}\label{eq:z}
    c_{\omega,i}(u_i) =
    \sum_{\tau(y) \land f(y) = g(x_i) \in\Gamma_{\omega,i}}
    [[\#y\,\tau(y) \land f(y) = g(u_i)]]^{\vG'}.
    \end{equation}
    Combining (\ref{eq:x}), (\ref{eq:y}) and (\ref{eq:z})
    yields our statement.
\end{proof}

\subsection{Constructing an Approximate Quantifier-Free Formula}\label{sec:constructingFO}

\Cref{thm:main} approximates a term $\#y\,\phi$ 
using functions $c_{\omega,i}(v)$ that assign each vertex a number.
In the following lemma, we round $c_{\omega,i}(v)$ into a finite
number of intervals 
(with a granularity depending on $\varepsilon$)
and extend the underlying structure with unary predicates
that encode for each vertex $v$ what interval $c_{\omega,i}(v)$ lies in.
Using these predicates, we build a quantifier-free formula $\phi$,
approximating $\#y\,\phi > N$.
This starts the quantifier elimination step we will later
use to build our approximation scheme.

\begin{lemma}\label{lem:2fo}
Let $\cal C \subseteq \fungraph{\sigma}$ be a class with
bounded expansion and $\varepsilon > 0$.  
One can compute for every quantifier-free formula $\phi(y\bx)$
with signature~$\sigma$,
a quantifier-free formula $\phi'(\bx)$
with signature $\rho\supseteq\sigma$ and the following property:

There exists another
class $\cal C' \subseteq \fungraph{\rho}$ with bounded expansion
and for every $\vG \in \cal C$ and $N \in \Z$ one can compute in time
$O(\norm{\vG})$ an expansion $\vG' \in \cal C'$ of $\vG$ such that for
every $\bu \in V(\vG)^{|\bx|}$:
\begin{itemize}
\item If $\vG' \models \phi'(\bu)$ then $\vG \models \#y\,\phi(y\bu) > N$.
\item If $\vG' \not\models \phi'(\bu)$ then $\vG \models \#y\,\phi(y\bu).
\le (1+\varepsilon)N$.
\end{itemize}
\end{lemma}

\begin{proof}
Let $\vG\in{\cal C}\subseteq\fungraph\sigma$ and $\phi(y\bx)$ be a
quantifier-free formula with signature~$\sigma$.
We can assume that $-1 \le N \le |\vG|$ because we
are counting the size of a vertex set.
By Theorem~\ref{thm:main} we can compute a signature $\rho^* \supseteq \sigma$,
a set of conjunctive clauses~$\Omega$ with signature $\rho^*$,
an expansion $\vG^*$ of $\vG$ and
functions $c_{\omega,i}$ such that for every $\bu\in V(\vG)^{|\bx|}$
$$
[[\#y\,\phi(y\bu)]]^\vG
\le \sum_{\omega\in\Omega}[[\omega(\bu)]]^{\vG^*
}\sum_{i=1}^{|\bx|}c_{\omega,i}(u_i)
\le (1+\varepsilon/2)
[[\#y\,\phi(y\bu)]]^\vG.
$$
The set $\Omega$ is such that for every $\bu\in\vG^{|\bx|}$ there is at most one $\omega \in \Omega$ with $\vG \models \omega(\bu)$.
Furthermore $\vG^*$ comes from a class $\cal C^* \subseteq \fungraph{\rho^*}$ with bounded expansion,
and $\vG^*$ as well as the functions $c_{\omega,i}$ can be computed in time linear in $\norm{\vG}$.
We define a step size $s = N\epsilon/2|\bx|$
and $c'_{\omega,i}(v) = s\lfloor c_{\omega,i}(v)/s\rfloor$.
We get $c'_{\omega,i}(v)$ if we round $c_{\omega,i}(v)$ down
to the next multiple of~$s$.
Therefore $c'_{\omega,i}(v) \le c_{\omega,i}(v) \le s + c'_{\omega,i}(v)$.
For every $\bu \in V(\vG)^{|\bx|}$
we get
$$
[[\cnt{y} \phi(y\bu)]]^{\vG} - |\bx| s
\le 
\sum_{\omega \in \Omega}
[[\omega(\bu)]]^{\vG^*}
\sum_{i=1}^{|\bx|}
c'_{\omega,i}(u_i)
\le 
(1+\epsilon/2)[[\cnt{y} \phi(y\bu)]]^{\vG}
$$
and because of $|\bx|s = N\epsilon/2$ also
\begin{align}\label{eqn:lala}
        \ip{\cnt{y}\phi(y\bu)}^{\vG} \le N &\implies
        \sum\limits_{\omega \in \Omega}
        \ip{\omega(\bu)}^{\vG^*}
        \sum\limits_{i=1}^{|\bx|}
        c'_{\omega,i}(u_i) \le (1+\epsilon/2)N,\\
        \label{eqn:lele}
        \ip{\cnt{y}\phi(y\bu)}^{\vG} > (1+\varepsilon) N &\implies
        \sum\limits_{\omega \in \Omega}
        \ip{\omega(\bu)}^{\vG^*}
        \sum\limits_{i=1}^{|\bx|}
        c'_{\omega,t}(u_i) > (1+\epsilon/2)N.
\end{align}
From now on, we will construct an expansion $\vG'$ of $\vG^*$ and a 
quantifier-free first-order formula
$\phi'(\bx)$ such that for $N'=(1+\epsilon/2)N$ and every $\bu\in
V(\vG)^{|\bx|}$
\begin{equation}\label{eqn:lulu}
        \vG'\models \phi'(\bu)
        \text{ iff }
        \sum\limits_{\omega \in \Omega}
        \ip{\omega(\bu)}^{\vG^*}
        \sum\limits_{i=1}^{|\bx|}
        c'_{\omega,i}(u_i) > N'.
\end{equation}
Using (\ref{eqn:lala}), (\ref{eqn:lele}) we see that this is sufficient
to prove this lemma.
We choose $l_{\max} \in \N$ independently of $N$ such that $sl_{\max} \ge N'$.
We define $\vG'$ to be the structure obtained from~$\vG^*$ by adding
the following unary relations.
For $\omega \in \Omega$,
$i \in \{1,\dots,|\bx|\}$, $l \in \{0,\dots,l_{\max}\}$ we add
the relations
\begin{align*}
R^l_{\omega,i} &= \{\,v \in V(\vG) \mid  c'_{\omega,i}(v) = sl\,\},\\
R^{\max}_{\omega,i} &= \{\,v \in V(\vG) \mid c'_{\omega,i}(v)>sl_{\max}\,\}.
\end{align*}
(For these relations we will not distinguish between the relations
themselves in the structure $\vG'$ and the corresponding relational
symbols in the signature.)
The number of such relations is independent of~$N$ and therefore
$\vG' \in \fungraph{\rho}$ for a signature $\rho \supseteq \rho^*$ 
whose size depends only on $\rho^*$ and~$\epsilon$.
Since $\vG'$ and $\vG^*$ have the same underlying graph,
$\vG'$ also belongs to a class $\cal C' \subseteq \fungraph{\rho}$
with bounded expansion.
Furthermore, the relations $R^l_{\omega,i}$ and $R^{\max}_{\omega,i}$ 
can be computed in time $O(\norm{\vG^*})$ from $\vG^*$
by evaluating and comparing each value $c'_{\omega,i}(v)$, $v \in V(\vG)$ in constant time.

The values of all functions $c'_{\omega,i}$ are always multiples of $s$.
Therefore, knowing the
relations $R^l_{\omega,t}$ for $0 \le l \le l_{\max}$ 
and $R^{\max}_{\omega,t}$
for a vertex $v \in V(\vG)$
gives us 
either the exact value of $c'_{\omega,t}(v)$ 
or indicates that $c'_{\omega,t}(v) > l_{\max}$.
With the help of those predicates we now define formulas. Let
$$
\phi'_1(\bx) = \bigvee_{\omega \in \Omega, i \in \{1,\dots,|\bx|\}}
\omega(\bx) \land R^{\max}_{\omega,i}(x_i).
$$
Observe that $\vG'\models\phi'_1(\bu)$ iff there exists
$\omega \in \Omega$, $i \in \{1,\dots,|\bx|\}$ with
$\ip{\omega(\bu)}^{\vG^*} c'_{\omega,i}(u_i) > sl_{\max} \ge N'$.
Let further
$$
L = \Bigl\{\,(l_1, \dots, l_{|\bx|}) \in \{0,\dots,l_{\max}\}^{|\bx|}
\;\Bigm|\;
\sum_{i=1}^{|\bx|} l_is > N\,\Bigr\}
$$
be the set of tuples whose sum is larger than $N$ when multiplied with~$s$.
We define
$$
\phi'_2(\bx) = 
\bigvee_{\omega \in \Omega}
\Bigl(
\omega(\bx)\; \land\!\!
\bigvee_{(l_1, \dots, l_{|\bx|})\in L} 
\bigwedge_{i=1}^{|\bx|} R^{l_i}_{\omega,i}(x_i)
\Bigr)
$$
and $\phi'(\bx)=\phi'_1(\bx) \lor \phi'_2(\bx)$.
We prove (\ref{eqn:lulu}) by a case distinction.
First, assume $\vG'\models\phi'_1(\bu)$.
Then also $\vG'\models\phi'(\bu)$ and, as shown above,
there exists $\omega \in \Omega$, $i \in \{1,\dots,|\bx|\}$ such that
$\ip{\omega(\bu)}^{\vG^*} c'_{\gamma,i}(u_i) > N'$.
This implies the right hand side of~(\ref{eqn:lulu}).
Next assume that $\vG'\not\models\phi'_1(\bu)$.
Then $\ip{\omega(\bu)}^{\vG^*} c'_{\omega,i}(u_i) \in \{0,\dots,l_{\max}\}$
for all $\omega \in \Omega$, $i \in \{1,\dots,|\bx|\}$.
Since $\Omega$ is such that $\vG' \models \omega$ for at most one $\omega \in \Omega$,
$$
        \vG'\models\phi'_2(\bu) \text{~~iff~~}
        \sum_{\omega \in \Omega} 
        \ip{\omega(\bu)}^{\vG^*} 
        \sum_{i=1}^{|\bx|} 
        c'_{\omega,i}(u_i) > N'.
$$
\end{proof}
For the special case $N = 0$,
\Cref{lem:2fo} immediately leads to the following result,
which was previously proven by Kazana and Segoufin
as a corner-stone of their quantifier elimination procedure~\cite{KazanaS2013}.

\begin{corollary}\label{col:eliminateExistential}
Let $\cal C \subseteq \fungraph{\sigma}$ be a class with
bounded expansion.
One can compute for every quantifier-free formula $\phi(y\bx)$
with signature~$\sigma$,
a quantifier-free formula $\phi'(\bx)$
with signature $\rho\supseteq\sigma$ and the following property:

There exists another
class $\cal C' \subseteq \fungraph{\rho}$ with bounded expansion
and for every $\vG \in \cal C$ one can compute in time
$O(\norm{\vG})$ an expansion $\vG' \in \cal C'$ of $\vG$ such that
$\vG \models \exists y \phi(y\bu)$ iff
$\vG' \models \phi'(\bu)$ holds for every $\bu \in V(\vG)^{|\bx|}$.
\end{corollary}

\subsection{Iterated Quantifier Elimination}\label{sec:iteratedElimination}




In this section we will finally construct our linear fpt model-checking
approximation scheme for graph classes with bounded expansion.
For some $\varepsilon > 0$, such a scheme 
either returns $1$ or $0$ or $\bot$.
The symbol $\bot$ stands for ``I do not know'' and may only be returned if the input
sentence is $(1+\varepsilon)$-unstable on the input graph.
We obtain such an algorithm by means of quantifier elimination:
For a given \FOX sentence in prenex normal form 
and input graph, we iteratively remove the innermost
quantifier of the sentence (while simultaneously computing expansions of the input graph) 
until no quantifiers are left.
Then we evaluate the remaining quantifier-free formula on the last expansion.

Using the previous \Cref{lem:2fo}, we can construct for a formula $\phi$ of the
form $\#y\,\psi > N$ (where $\psi$ is quantifier-free)
a new quantifier-free formula $\phi'$ that approximates $\phi$.
This effectively removes one counting quantifier.
We will use this as the main building block of our quantifier elimination procedure.
But since $\phi'$ only approximates $\phi$, we need to be careful.
If the counting term $\#y\,\psi$ evaluates to something greater than $N$ but
not greater than $(1+\varepsilon)N$
then $\phi$ and $\phi'$ may give a different answer.
Since an approximation scheme is not allowed to give the wrong answer,
we cannot simply replace $\phi$ with $\phi'$.

Looking at \Cref{lem:2fo}, we notice that $\phi'$ is, in a sense,
an ``underapproximation'' of $\phi$:
If $\phi'$ says yes, then we know for sure that $\phi$ is satisfied,
but sometimes $\phi'$ may say no, even though $\phi$ is still satisfied.
As we will see soon, 
we could also construct $\phi'$ to be an ``overapproximation'':
If $\phi'$ says no, then we know for sure that $\phi$ is not satisfied,
but sometimes $\phi'$ may say yes, even though $\phi$ is unsatisfied.

In the following, we will define over- and underapproximations
of a \FOX-formula $\phi$, denoted by $\phi_+',\phi_-'$, respectively.
If the over- and underapproximation agree then we know for certain whether $\phi$ is satisfied.
But if they do not agree, then $\phi$ is $(1+\varepsilon)$-unstable and 
our approximation scheme is allowed to return ``I do not know.''
We call a tuple consisting of an over- and underapproximation together with an
expansion of the input graph an \emph{$\varepsilon$-approximation} (see \Cref{def:approx}).
Then in \Cref{lem:plumbing} and \ref{lem:oneelimination}
we show how to obtain an $\varepsilon$-approximation from the results of the previous sections
and how to chain $\varepsilon$-approximations in a meaningful way.
These two lemmas are quite technical, but ultimately unexciting.
At last, in \Cref{thm:main} we perform the quantifier elimination procedure.
We will iteratively remove the innermost
quantifier while expanding the input graph and
maintaining an over- and underapproximation of the original \FOX sentence.
In the end, we evaluate both approximations and only return $\bot$ if they disagree.


\begin{definition}\label{def:approx}
Let $\phi(\bx)$ be a \FOX-formula, $\varepsilon~>~0$, $\vG$ be a functional structure, 
$\vG'$ be an expansion of $\vG$, and $\phi_+'(\bx),\phi_-'(\bx)$ be quantifier-free formulas.
We say
$(\vG',\phi_+',\phi_-')$ is an \emph{$\varepsilon$-approximation} of $(\vG,\phi)$ 
if there are formulas $\tilde\phi_+$, $\tilde\phi_-$ that are
$(1+\varepsilon)$-similar to $\phi$
and for every $\bu \in V(\vG)^{|\bx|}$ the following holds.
\begin{enumerate}
    \item If $\vG' \models \phi_-'(\bu)$ then $\vG \models \phi(\bu)$.
    \item If $\vG' \not\models \phi_-'(\bu)$ then $\vG \not\models \tilde\phi_-(\bu)$.
    \item If $\vG' \not\models \phi_+'(\bu)$ then $\vG \not\models \phi(\bu)$.
    \item If $\vG' \models \phi_+'(\bu)$ then $\vG \models \tilde\phi_+(\bu)$.
\end{enumerate}
\end{definition}

It follows a technical observation about sufficient conditions to
obtain an $\varepsilon$-approxi\-mation
of a formula of the form $\#y\phi > N$, assuming we already have an approximation of $\phi$.

\begin{lemma}\label{lem:plumbing}
    Let $\phi(y\bx)$ be a \FOX-formula, $\vG$ be a functional structure, $\varepsilon > 0$, and $N \in \Z$.
    Assume $(\vG', \phi_+',\phi_-')$ is an $\varepsilon$-approximation of $(\vG,\phi)$.
    A tuple $(\vG'', \phi_+'',\phi_-'')$ is an
    $\varepsilon$-approximation of $(\vG,\#y\,\phi > N)$
    if $\vG''$ is an expansion of $\vG'$,
    the formulas $\phi_+''(\bx)$, $\phi_-''(\bx)$ 
    are quantifier-free
    and for every $\bu \in V(\vG)^{|\bx|}$ the following holds.
    \begin{enumerate}
        \item If $\vG'' \models \phi_-''(\bu)$ 
            then $\vG' \models \#y\,\phi_-'(y\bu) > N$.
        \item If $\vG'' \not\models \phi_{-}''(\bu)$
            then $\vG' \models \cnt{y}\phi'_-(y\bu) \le (1+\epsilon)N$.
        \item If $\vG'' \not\models \phi_+''(\bu)$ 
            then $\vG' \models \#y\,\phi_+'(y\bu) \le N$.
        \item If $\vG'' \models \phi_{+}''(\bu)$
            then $\vG' \models \cnt{y}\phi'_+(y\bu) > (1-\epsilon)N$.
    \end{enumerate}
\end{lemma}
\begin{proof}
    We show that $(\vG'', \phi_+'',\phi_-'')$ satisfies 1.\ to 4.\ of \Cref{def:approx} one by one.
    Let $\bu \in V(\vG)^{|\bx|}$.
    \begin{enumerate}
        \item 
            We need to prove that if $\vG'' \models \phi_-''(\bu)$ then $\vG \models \#y\,\phi(y\bu) > N$.
            $(\vG', \phi_+',\phi_-')$ is an $\varepsilon$-approximation of $(\vG,\phi)$,
            and therefore by 1.\ from \Cref{def:approx} for every $v \in V(\vG)$
            with $\vG' \models \phi_-'(v\bu)$ we get $\vG \models \phi(v\bu)$.
            If $\vG'' \models \phi_-''(\bu)$
            then according to the statement of this lemma $\vG' \models \#y\,\phi_-'(y\bu) > N$
            and by the previous argument also $\vG \models \#y\,\phi(y\bu) > N$.
        \item 
            According to 2.\ from \Cref{def:approx}, there exists a formula $\tilde\phi_-$ that is $(1+\varepsilon)$-similar to $\phi$ such that
	    $\vG \not\models \tilde\phi_-(v\bu)$ holds
	    for all $v \in V(\vG)$
            with $\vG' \not\models \phi_-'(v\bu)$.
            By \Cref{def:similar}, the formula $\cnt{y}\tilde\phi_- > (1+\epsilon)N$ is $(1+\varepsilon)$-similar to $\#y\,\phi > N$.
            If $\vG'' \not\models \phi_-''(\bu)$
            then according to the statement of this lemma
            $\vG' \not\models \cnt{y}\phi'_-(y\bu) > (1+\epsilon)N$
            and as discussed above
            $\vG' \not\models \cnt{y}\tilde\phi_-(y\bu) > (1+\epsilon)N$.
        \item Similar to 1.
        \item Similar to 2.
    \end{enumerate}
\end{proof}

The next lemma describes our main quantifier elimination step.
Assuming we already have an approximation of a formula $\phi$,
it gives us an approximation of $\#y\phi~>~N$.

\begin{lemma}\label{lem:oneelimination}
Let $\cal C' \subseteq \fungraph{\sigma}$ be a class with bounded expansion and
$\varepsilon > 0$.  One can compute for given quantifier-free formulas
$\phi_+'(y\bx)$, $\phi_-'(y\bx)$ with signature~$\sigma$ two quantifier-free
formulas $\phi_+''(\bx)$, $\phi_-''(\bx)$ with signature $\rho\supseteq\sigma$
and the following property:

There exists another class $\cal C'' \subseteq \fungraph{\rho}$ with bounded
expansion and for every $\vG' \in \cal C'$ and $N \in \Z$ one can compute in
time $O(\norm{\vG'})$ an expansion $\vG'' \in \cal C''$ of $\vG'$ such that for every
pair $(\vG,\phi(y\bx))$:

If $(\vG', \phi_+',\phi_-')$ is an
$\varepsilon$-approximation of $(\vG,\phi)$ then $(\vG'', \phi_+'',\phi_-'')$
is an $\varepsilon$-approxi\-mation of $(\vG,\#y\,\phi > N)$.
\end{lemma}

\begin{proof}
    We use \Cref{lem:2fo} to construct a formula $\phi_-''(\bx)$ and an expansion
    $\vG^*$ of $\vG'$ such that 
    for every $\bu \in V(\vG)^{|\bx|}$ the following holds.
    \begin{itemize}
        \item 
            If $\vG^* \models \phi_-''(\bu)$
            then $\vG' \models \#y\,\phi_-'(y\bu) > N$.
        \item 
            If $\vG^* \not\models \phi_-''(\bu)$
            then $\vG' \models \#y\,\phi_-'(y\bu) \leq (1+\varepsilon) N$.
    \end{itemize}
    Then we use \Cref{lem:2fo} again with
    $M = \floor{(1-\varepsilon) N}$ to construct 
    a formula $\phi''_+$ and an expansion $\vG'$ of $\vG^*$ such that 
    for every $\bu \in V(\vG)^{|\bx|}$ we have the following.
    \begin{itemize}
        \item 
            If $\vG'' \models \phi_+''(\bu)$
            then $\vG^* \models \#y\,\phi_+'(y\bu) > M$.
        \item 
            If $\vG'' \not\models \phi_+''(\bu)$
            then $\vG^* \models \#y\,\phi_+'(y\bu) \leq (1+\varepsilon) M$.
    \end{itemize}
    Since $\vG''$ is an expansion of $\vG^*$,
    $\vG^*$ is an expansion of $\vG'$ and $(1+\varepsilon) M \le N$
    these four statements are satisfied for every $\bu \in V(\vG)^{|\bx|}$.
    \begin{enumerate}
        \item If $\vG'' \models \phi_-''(\bu)$ 
            then $\vG^* \models \phi_-''(\bu)$,
            and then $\vG' \models \#y\,\phi_-'(y\bu) > N$.
        \item If $\vG'' \not\models \phi_{-}''(\bu)$
            then $\vG^* \not\models \phi_{-}''(\bu)$,
            and then $\vG' \models \cnt{y}\phi'_-(y\bu) \le (1+\epsilon)N$.
        \item If $\vG'' \not\models \phi_+''(\bu)$ 
            then $\vG^* \models \#y\,\phi_+'(y\bu) \le (1+\varepsilon) M$,
            and then $\vG' \models \#y\,\phi_+'(y\bu) \le N$.
        \item If $\vG'' \models \phi_{+}''(\bu)$
            then $\vG^* \models \cnt{y}\phi'_+(y\bu) > M$,
            and then $\vG' \models \cnt{y}\phi'_+(y\bu) > (1-\epsilon)N$.
    \end{enumerate}
    Thus if $(\vG', \phi_+',\phi_-')$ is an $\varepsilon$-approximation of
    $(\vG,\phi)$ 
    then by \Cref{lem:plumbing}, $(\vG'', \phi_+'',\phi_-'')$ is an
    $\varepsilon$-approximation of $(\vG,\#y\,\phi > N)$.
\end{proof}

It will be convenient to convert \FOX-formulas
into the following normal form.

\begin{definition}
    We say a \FOX sentence $\phi$ is in \emph{counting prenex normal form}
    if it contains no $\exists$-quantifiers, no $\forall$-quantifiers and
    all subformulas of the form $\phi_1 \lor \phi_2$ or $\phi_1 \land \phi_2$
    are such that $\phi_1$ and $\phi_2$ are quantifier-free.
\end{definition}

\begin{lemma}\label{lem:countingnormalform}
    For every \FOX sentence $\phi$ one can compute a \FOX sentence
    $\phi'$ in counting prenex normal form such that $\phi \equiv \phi'$ and
    for structure $\vG$ and $\lambda > 1$, $\phi$ is $\lambda$-stable on $\vG$
    iff $\phi'$ is $\lambda$-stable on $\vG$.
\end{lemma}
\begin{proof}
    Subformulas of the form $\exists y \psi$
    can be substituted by $\cnt{y} \psi>0$.
    The remaining construction of $\phi'$ is analogous to
    the construction for a prenex normal form of first-order formulas.
    For example, subformulas $\#y\, \phi_1 > N \wedge \phi_2$ are replaced with
    $\#y\, \phi_1 \wedge \phi_2 > N$, renaming variables if necessary.
    It follows immediately that $\phi$ is $\lambda$-stable if and only if $\phi'$ is $\lambda$-stable.
\end{proof}

A formula in counting prenex normal form with at least one quantifier is either of the form
$\#y\,\phi > N$ or $\#y\,\phi \le N$, where $\phi$ is again in counting prenex normal form.
Using \Cref{lem:oneelimination}, we can approximate
formulas of the former form.
The following simple observation additionally gives us 
approximations of the latter.

\begin{observation}\label{obs:negation}
    If $(\vG',\phi_+,\phi_-)$ is an $\varepsilon$-approximation of $(\vG,\phi)$
    then $(\vG',\neg\phi_-,\neg\phi_+)$ is an $\varepsilon$-approximation of $(\vG,\neg\phi)$.
\end{observation}

We are now ready to prove our main result:
an approximation scheme as described in \Cref{def:modelcheckingblub}.

\begin{customthm}{\ref{thm:approxScheme}}
There is a linear fpt model-checking approximation scheme for
\FOX on labeled graph classes with bounded expansion.
\end{customthm}

\begin{proof}
As discussed in \Cref{sec:functionalRepresentations},
we can assume that the input to our approximation scheme is a functional structure $\vG$,
taken from a class $\cal C$ with bounded expansion
and a functional sentence $\phi$ (as well as $\varepsilon > 0)$.
By \Cref{lem:countingnormalform},
we can further assume that our input sentence $\phi$
is given in counting prenex form.

The main idea behind this proof is to iteratively remove the innermost
counting quantifier using \Cref{lem:oneelimination} until none are left.
Then we can easily evaluate the remaining quantifier-free formula.
For each quantifier, we will replace the input structure with an expansion.
For technical reasons (\Cref{lem:oneelimination} requires a non-empty tuple $\bx$ of free variables), 
we add an unused free variable $x$ to $\phi$, and call the formula
$\phi(x)$. The fact whether $\phi(x)$ is satisfied in $\vG$ is independent
of the assignment to $x$.

We define a sequence of increasing subformulas $\phi^0,\dots,\phi^l$ of
$\phi$ with $\phi^l = \phi$ as follows:
$\phi^0$ is the maximal quantifier-free subformula of $\phi$
and $\phi^{i+1}$ is either of the form $\#\,y\phi^i > N$ or $\#\,y\phi^i \le N$ for some $N \in \Z$.

In the following, we will construct an $\varepsilon$-approximation
$(\vG^i,\phi^i_+,\phi^i_-)$ of $(\vG,\phi^i)$ for every $0 \le i \le l$.
Since $\phi^0$ is quantifier-free, by \Cref{def:approx}
$(\vG,\phi^0,\phi^0)$ is an $\varepsilon$-approximation of $(\vG,\phi^0)$
and we set set $(\vG^0,\phi^0_+,\phi^0_-) = (\vG,\phi^0,\phi^0)$.
Assume we already have an $\varepsilon$-approximation $(\vG^i,\phi^i_+,\phi^i_-)$ of $(\vG,\phi^i)$.
Then \Cref{lem:oneelimination} and \Cref{obs:negation} 
give us an $\varepsilon$-approximation
$(\vG^{i+1},\phi^{i+1}_+,\phi^{i+1}_-)$ of $(\vG,\phi^{i+1})$ as well.
Each formula $\phi^{i+1}_+,\phi^{i+1}_-$ is constructed independently
of the structure and depends only on $\varepsilon$ and $\phi$.
Furthermore, each expansion $\vG^{i+1}$ can be constructed from $\vG^{i}$ in linear time.
In the end, we have an $\varepsilon$-approximation
$(\vG^l,\phi^l_+,\phi^l_-)$ of $(\vG,\phi)$.
Since the formulas $\phi^l_+(x),\phi^l_-(x)$ are quantifier-free,
we can easily evaluate them on $\vG^l$ (since $x$ has no purpose in $\phi$, we assign an arbitrary vertex to $x$).
We can distinguish three outcomes:
\begin{itemize}
    \item If $\vG^l \models \phi^l_-$ then by \Cref{def:approx}, $\vG \models \phi$. We return 1.
    \item If $\vG^l \not\models \phi^l_+$ then by \Cref{def:approx}, $\vG \not\models \phi$. We return 0.
    \item If $\vG^l \not\models \phi^l_-$ and $\vG^l \models \phi^l_+$ then by \Cref{def:approx} and \ref{def:stable},
        $\phi$ is $(1+\varepsilon)$-unstable on $\vG$. We return~$\bot$.
\end{itemize}

Thus, for every $\varepsilon > 0$ we have given an algorithm
which takes as input a sentence $\phi$ and a structure $\vG \in \cal C$,
runs in time $f(|\phi|,\varepsilon)\norm{\vG}$ for some function $f$,
and whose output satisfies the criteria of a 
linear fpt model-checking approximation scheme,
as formulated in \Cref{def:modelcheckingblub}.

But \Cref{def:modelcheckingblub} further
requires that there is a \emph{single} algorithm (independent of $\varepsilon$)
taking $\varepsilon$, $\phi$, and $\vG$ as input.
We only presented one algorithm \emph{for each} $\varepsilon > 0$.
Our proofs are structured such that
for each $\varepsilon > 0$, we use the algorithms in
Lemma~
\ref{lem:aug},
\ref{lem:aug2},
\ref{lem:apx},
\ref{lem:2fo},
\ref{lem:oneelimination}
and \Cref{thm:main} as subroutines.
After further inspection we notice that these algorithms itself
can be easily computed from $\varepsilon$.
This gives us a single algorithm
running in time $f(|\phi|,\varepsilon)\norm{\vG}$
and therefore a linear fpt  model-checking approximation scheme.
\end{proof}

\section{Exact Counting and Optimization}\label{sec:exact}

In the previous section, we just presented a linear fpt model-checking approximation scheme for \FOX on graph classes with bounded expansion.
Now we may ask ourselves if this can be extended into an (exact) model-checking algorithm.
While we will show in \Cref{sec:hardness} that the parameterized model-checking
problem for \FOX is already hard on very simple structures,
there may still be some fragments of \FOX where we can efficiently solve the
problem on graph classes with bounded expansion.
In this section we solve an optimization problem that leads to such a fragment.
For a given first-order
formula $\phi(y\bx)$ and graph $G$ from a bounded expansion class we find 
in linear fpt time an assignment for
$\bx$ which maximizes (or minimizes) the value of $\#y\,\phi(y\bx)$ in $G$.
This is presented in \Cref{thm:optimization} at the end of this section.
From \Cref{thm:optimization}, we obtain the following 
corollary about a fragment of \FOX that admits efficient model-checking.
\begin{corollary}
    Let $\cal C$ be a graph class with bounded expansion.
    There exists a function $f$ such that for a given
    graph $G \in \cal C$, $N \in \Z$ and first-order formula $\phi(y\bx)$
    one can compute in time
    $f(|\phi|)\norm{G}$ whether
    $G \models \exists \bx \#y\, \phi(y\bx) > N$.
\end{corollary}

The formula
$\exists x_1 \dots \exists x_k \, \# y \; \bigl( \bigvee_i E(y,x_i) \lor y = x_i \bigr) \ge N$
expresses whether there are $k$ vertices dominating $N$ vertices.
We therefore obtain an efficient algorithm for partial dominating set as a special case
of the previous corollary.
\begin{customcol}{\ref{col:partialDomSet}}
    Partial dominating set can be solved in linear fpt time
    on graph classes with bounded expansion.
\end{customcol}



\subsection{Finding Optimal Assignments}

Again we will prove our results in this section using functional
representations of graph classes with bounded expansion.
We can recycle many steps from \Cref{sec:approxModelChecking}.
We will only go into detail for those parts that differ significantly from
\Cref{sec:approxModelChecking}.
Key ingredients in this section are the inclusion-exclusion principle
and low tree-depth colorings.

We start with proving a counterpart to \Cref{lem:apx}.
In \Cref{lem:apx}, we break down a counting term
into a sum of positive summands that approximate it.
Here, we use the inclusion-exclusion principle to
break it down into an (exact) sum containing both positive and negative
summands.

\begin{lemma}\label{lem:apxExact}
Let $\cal C \subseteq \fungraph{\sigma}$ be a class with bounded expansion.
One can compute for every first-order formula $\phi(y\bx)$
with signature $\sigma$ 
a signature $\rho\supseteq\sigma$ and
a set $\Omega$
with the following properties:

\begin{itemize}
\item
    The set $\Omega$ contains pairs $(\mu,\omega(y\bx))$ where
    $\mu \in \{1,-1\}$ and 
    $\omega(y\bx)  \in \funformula{2}{\rho}$
    is of the form $\tau(y)\land\psi(\bx)\land f(y)=g(x_i)$.
\item
    There exists a class $\cal C' \subseteq \fungraph{\rho}$ with bounded expansion
    and for every $\vG \in \cal C$ one can compute in time $O(\norm{\vG})$
    an expansion $\vG' \in \cal C'$ of $\vG$ 
    such that for every $\bu \in V(\vG)^{|\bx|}$
    $$
    [[\#y\, \phi(y\bu)]]^\vG = 
    \sum_{(\mu,\omega) \in \Omega} \mu [[\#y\, \omega(y\bu)]]^{\vG'}.
    $$
\end{itemize}
\end{lemma}

\begin{proof}

We can assume $\phi(y\bx)$ to be given in prenex normal form and
use the quantifier elimination as presented in
\Cref{col:eliminateExistential} (or by Kazana and Segoufin~\cite{KazanaS2013})
to iteratively remove the innermost quantifier from $\phi(y\bx)$.
While doing so, we replace $\vG$ with an expansion of $\vG$
that still comes from a graph class with bounded expansion and can be computed in linear time.
We repeat this procedure until $\phi(y\bx)$ is quantifier-free.
As mentioned in the beginning of \Cref{lem:apx}'s proof,
we can do some further simple
modifications on our input formula $\phi(y\bx)$ and structure $\vG$
and assume without loss of generality
that $\phi(y\bx)$ is a quantifier-free formula with functional depth one
and that there exists a function symbol $\fapx \in \fsig\sigma$ with
$\fapx^{\vG}(v) = \fapx^{\vG}(v')$ for every $v,v' \in V(\vG)$.

Let $\vG'$ be the 1-transitive fraternal augmentation of $\vG$.
As discussed in \Cref{sec:functionalRepresentations},
it can be computed in time $O(\norm{\vG})$.
Also, \Cref{cor:funcBE} states that $\vG'$ belongs to some class
with bounded expansion $\cal C' \subseteq \fungraph{\rho}$ depending only on $\cal C$.
Let $\bu \in V(G)^{|\bx|}$.
Using \Cref{lem:canonical-formula}, we obtain a mutually exclusive set 
$\Omega' \subseteq \C(|\bx|,\sigma,\rho)$ of canonical conjunctive clauses
such that for every $v\bu \in V(\vG)^{|y\bx|}$ 
$$
\vG\models\phi(v\bu) \text{~~iff~~}
\vG'\models\phi(v\bu)\land \fapx(v) = \fapx(u_1) \text{~~iff~~}
\vG'\models\omega(v\bu)\text{ for some }\omega\in\Omega'.
$$
Let $\Omega = \{(1,\omega) \mid \omega \in \Omega'\}$.
Since at most one formula $\omega \in \Omega$ is true at a time ($\Omega$ is mutually exclusive),
this implies
\begin{equation}\label{eq:lambdasum}
[[\#y\, \phi(y\bu)]]^\vG = 
\sum_{(\mu,\omega) \in \Omega} \mu [[\#y\, \omega(y\bu)]]^{\vG'}.
\end{equation}

Let $(\mu,\omega) \in \Omega$.
The canonical conjunctive clause $\omega(y\bx)$ is of the form $\tau(y) \land \psi(\bx) \land \Delta^=(y\bx) \land \Delta^{\neq}(y\bx)$.
This lemma requires that $\Delta^=(y\bx)$ contains exactly one literal and $\Delta^{\neq}(y\bx)$ is empty,
which is not yet the case.
We will gradually modify $\Omega$ until these requirements are met.
We start with making sure that $\Delta^{\neq}(y\bx)$ is empty.

If $\Delta^{\neq}(y\bx)$ is non-empty, we can write it as
$f(y) \neq g(x_i) \land \Delta'^{\neq}(y\bx)$.
By first ignoring the literal $f(y) \neq g(x_i)$
and then subtracting what we counted too much we get
\begin{align}\label{eq:inclexcl}
\mu [[\#y\, \omega(y\bx)]]^{\vG'}
=& \mu [[\#y\, \tau(y) \land \psi(\bx) \land \Delta^=(y\bx) \land \Delta^{'\neq}(y\bx)]]^{\vG'} \nonumber\\
-& \mu [[\#y\, \tau(y) \land \psi(\bx) \land \Delta^=(y\bx) \land f(y) = g(x_i) \land \Delta^{'\neq}(y\bx)]]^{\vG'}.
\end{align}
We remove $(\mu,\omega)$ from $\Omega$
and add two new entries with canonical conjunctive clauses as
in~(\ref{eq:inclexcl}) such that $\Omega$ still satisfies~(\ref{eq:lambdasum}).
Both newly introduced formulas contain one negative literal less
in their respective set $\Delta^{\neq}(y\bx)$.
We perform this procedure on $\Omega$ until no longer possible.

Consider an element $(\mu,\omega) \in \Omega$.
The clause $\omega(\bx)$ is now of the form $\tau(y) \land \psi(\bx) \land \Delta^=(y\bx)$.
We apply \Cref{lem:step1} to either conclude that $\omega$
is $\sigma$-$\rho$-unsatisfiable or obtain an equivalent
formula of the form $\tau(y) \land \psi(\bx) \land f(y) = g(x_i)$.
If $\omega$ is $\sigma$-$\rho$-unsatisfiable then (since $\vG'$ is a $\sigma$-$\rho$-expansion)
$[[\#y\, \omega(y\bx)]]^{\vG'} = 0$ and we can remove it.
Otherwise, we replace it with $\tau(y) \land \psi(\bx) \land f(y) = g(x_i)$.
$\Omega$ still satisfies (\ref{eq:lambdasum}).
We repeat this for every pair in~$\Omega$.
Then $\Omega$ is of the desired form.

\end{proof}

Next, we present a counterpart to \Cref{thm:main}.
The proof of this theorem is very similar to the proof
of \Cref{thm:main}, and therefore omitted.
The only real difference is that here, we call the algorithms
from \Cref{lem:apxExact} instead of \Cref{lem:apx} as a subroutine.

\begin{theorem}\label{thm:mainExact}
Let $\cal C \subseteq \fungraph{\sigma}$ be a class with bounded expansion.
One can compute for every 
first-order formula
$\phi(y\bx)$
with signature $\sigma$
a set of conjunctive clauses $\Omega$ 
with free variables $\bx$,
and signature $\rho \supseteq \sigma$ that satisfies the following property:

There exists a class $\cal C' \subseteq \fungraph{\rho}$ with bounded expansion
such that for every $\vG \in \cal C$ one can compute in time $O(\norm{\vG})$
an expansion $\vG' \in \cal C'$ of $\vG$ 
and functions $c_{\omega,i}(v) \colon V(\vG) \to \Z$ for $\omega \in \Omega$ and $i \in \{1,\dots,|\bx|\}$
with $c_{\omega,i}(v) = O(|\vG|)$
such that for every $\bu \in V(\vG)^{|\bx|}$
there exists exactly one formula $\omega \in \Omega$
with $\vG' \models \omega(\bu)$. For this formula 
$$
[[\#y\,\phi(y\bu)]]^\vG = \sum_{i=1}^{|\bx|}c_{\omega,i}(u_i).
$$
\end{theorem}

This brings us to the main result of this section.
Using \Cref{thm:mainExact}, we can break down a counting term
into a sum of values that depend only on a single vertex.
We can consider these values the ``weights'' of a vertex.
In \Cref{thm:optimization}, we have to find an assignment
to a quantifier-free first-order formula that maximizes (or minimizes) the weights.
We use low tree-depth colorings to break the input graph down
into a small number of subgraphs with bounded tree-depth
and then use an optimization variant of Courcelle's theorem
by Courcelle, Makowsky and Rotics~\cite{CourcelleMR2000} to solve our problem.

\begin{customthm}{\ref{thm:optimization}}
Let $\cal C$ be a labeled graph class with bounded expansion.
There exists a function $f$ such that for a given
graph $G \in \cal C$ and first-order formula $\phi(y\bx)$
one can compute in time
$f(|\phi|)\norm{G}$ a tuple
$\bu^* \in V(G)^{|\bx|}$ such that
$$
[[\#y\, \phi(y\bu^*)]]^{G} = 
\underset{\bu \in V(G)^{|\bx|}}{\textnormal{opt}} [[\#y\, \phi(y\bu)]]^{G},
$$
where \textnormal{opt} is either $\min$ or $\max$.
\end{customthm}

\begin{proof}
As discussed in \Cref{sec:functionalRepresentations},
we can assume that the input $G$, $\phi$ is represented by a 
functional representation $\vG$ and a functional formula $\vec \phi$.
We use \Cref{thm:mainExact} to compute
a set $\Omega$,
an expansion $\vG'$ of $\vG$ 
and functions $c_{\omega,i}(v)$ with $c_{\omega,i}(v) = O(|\vG|)$
such that for every $\bu \in V(\vG)^{|\bx|}$
$$
[[\#y\,\phi(y\bu)]]^G =
[[\#y\,\vec\phi(y\bu)]]^\vG =
\sum_{i=1}^{|\bx|}c_{\omega,i}(u_i),
$$
where $\omega \in \Omega$ is the unique formula with
$\vG' \models \omega(\bu)$.
Assume for now that we can compute for a given $\omega \in \Omega$
a tuple $\bu^* \in V(G)^{|\bx|}$ such that
\begin{equation}\label{eq:newopti}
\sum_{i=1}^{|\bx|}c_{\omega,i}(u^*_i) = 
\textnormal{opt}\Bigl\{ \sum_{i=1}^{|\bx|}c_{\omega,i}(u_i) \Bigm|
\bu \in V(\vG)^{|\bx|}, \vG' \models \omega(\bu) \Bigr\}.
\end{equation}
Then we could cycle through all $\omega \in \Omega$,
compute a solution $\bu^*$ satisfying (\ref{eq:newopti}),
and return the optimal value for $\bu^*$ among all of them.
This gives us a solution to our original optimization problem.
Thus, from now on, we will concentrate on one formula $\omega \in \Omega$
and solve the optimization problem~(\ref{eq:newopti}).

It will now be easier for us to work with relational instead of functional
structures.
We can transform $\vG'$ into a relational structure $G'$
with the same universe via standard methods:
The unary relations are preserved.
Additionally, for every function $f$ we add a binary relation $E_f$ with
$E^{G'}_f = \{\,(v,f^{\vG'}(v)) \mid v \in V(\vG')\,\}$.
The resulting structure is a directed graph with self-loops, vertex- and edge-labels
and belongs to a graph class with bounded expansion.
We further construct a relational conjunctive clause $\omega'(\bx\bz)$ such that
$\vG' \models \omega(\bu)$ iff $G' \models \exists \bz \omega'(\bu\bz)$
for every $\bu \in V(\vG)^{|\bx|}$.
This can be done for example by iteratively replacing
every atom of the form $f(x)$
with some newly introduced variable $z$ and adding a literal $E_f(x,z)$.

Graph classes with bounded expansion can be characterized
by so called \emph{low tree-depth colorings}~\cite{NesetrilM2012}.
This means there exists a function $\chi\colon\N\to\N$
(depending only on the graph class) such that
we can color the vertices of $G'$ with $\chi(|\bx\bz|)$ colors and
every subset of $|\bx\bz|$ many colors induces a graph with tree-depth at most 
$|\bx\bz|$.
Let $\cal H$ be the set of all graphs obtained from $G'$ by inducing
it on some set of $|\bx\bz|$ many colors of this low tree-depth coloring.
The size of $\cal H$ is bounded by a constant independent of~$G'$.

For every $\bu \in V(\vG')^{|\bx|}$ with $\vG' \models \omega(\bu)$
there exists $H \in \cal H$ such that $\bu \in V(H)^{|\bx|}$
and $H \models \exists \bz \omega'(\bu\bz)$.
In order to optimize (\ref{eq:newopti}), 
it is therefore sufficient to look at every graph $H \in \cal H$
and compute $\bu^* \in V(H)^{|\bx|}$ such that
\begin{equation}\label{eq:optitd}
    \sum_{i=1}^{|\bx|}c_{\omega,i}(u^*_i) = 
    \textnormal{opt}\Bigl\{\,\sum_{i=1}^{|\bx|}c_{\omega,i}(u_i) \Bigm|
    \bu \in V(H)^{|\bx|}, H \models \exists \bz \omega'(\bu\bz)\,\Bigr\},
\end{equation}
and then return the best value found for~$\bu^*$.
The input graph $H$ to the optimization problem (\ref{eq:optitd})
comes from a graph class with bounded tree-depth.
Using Courcelle's theorem \cite{Courcelle1990}
one can solve a wide range of problems on these graphs in fpt time.
Since we want to solve an optimization problem we require an extension
of the original theorem, as presented in \cite{CourcelleMR2000}.
There, the authors define \textit{LinEMSOL}~\cite{CourcelleMR2000}
as an extension of monadic second order logic allowing one to search for
sets of vertices that are optimal with respect to a linear evaluation
function.

Our linear evaluation function is
$\sum_{i=1}^{|\bx|}c_{\omega,i}(u_i)$
and our formula $\omega'(\bx\bz)$ clearly lies in monadic second order logic.
The only problem is that \cite{CourcelleMR2000} defines \textit{LinEMSOL}
for undirected graphs with vertex labels,
while $H$ is a directed graph with edge labels, vertex labels and loops.
Using standard techniques (subdivisions of edges, adding new vertex labels)
we can transform $H$ into an undirected graphs with vertex labels
and then find a solution $\bu^*$ to our optimization problem (\ref{eq:optitd})
in fpt time.
\end{proof}

\section{Hardness Results}\label{sec:hardness}
\newcommand{\approxratio}[1]{2^{\sqrt{\log(#1)}}}

In this section we show that in some sense, the
linear fpt model-checking approximation scheme for \FOX presented in this work is optimal:
We show that the (exact) model-checking problem for 
\FOX is already AW[$*$]-hard on trees of depth 4,
and therefore most likely not in FPT (\Cref{lem:exacthard}).
We further show that (unless AW[$*$] $\subseteq$ FPT) 
certain other fragments of \FOCless 
that are slightly stronger than \FOX do not admit a 
model-checking approximation algorithm on the class of all trees with depth~9 ---
even if we only want an ``approximation ratio'' of $\approxratio{n}$ (\Cref{lem:approxhard}).
This means we cannot hope for model-checking approximation schemes
on graph classes with bounded expansion if we allow for example 
comparison of non-atomic counting terms (e.g., $1 \cdot \cnt{y}\phi_1 > 1 \cdot \cnt{z}\phi_2$)
or multiplication (e.g., $\cnt{y}\phi_1 \cdot \cnt{z}\phi_2 > N$).

The ideas behind the proofs of this section
are considerably simpler than those of the remaining paper,
but unfortunately they involve very technical constructions.
We reduce from the AW[$*$]-complete
parameterized first-order model-checking problem on the class of all graphs.
The reduction showing hardness of the model-checking problem for \FOX
is very similar to~\cite[Theorem 4.1]{GroheS2018}.  
We encode an arbitrarily graph as a tree of depth 4
such that the underlying graph can be recovered using formulas in \FOX.

As this reduction is quite unstable (slight changes in the number-constants of the formula lead to the wrong answer),
we have to use a more sophisticated reduction in \Cref{sec:foxextensionhard} to show hardness of approximation.
For a graph with $n$ vertices, the information ``$u$ is adjacent to $v$'' can be written down using $O(\log(n))$ bits
by assuming that the vertices of the graph are the numbers from $1$ to $n$
and then giving a binary encoding of $u$ and $v$.
Since this encoding has length $O(\log(n))$, we only need $O(\log\log(n))$
bits to write down the information ``the $i$th bit of some adjacency-encoding
is~1.''
We store this information (for every $i$ and every edge) in a tree of bounded depth.
We can then recover the original graph in a very robust manner.
The counting terms will only need to
distinguish between $O(\log\log(n))$ different possible values
and therefore even larger perturbations in the constants of the formula
will still lead to the correct answer.

\subsection{Hardness of Model-Checking for \BOLDFOX}\label{sec:foxhard}

\begin{lemma}\label{lem:exacthard}
The model-checking problem for \FOX on the class of all trees of depth~4
is \textnormal{AW[$*$]}-complete.
\end{lemma}

\begin{proof}
The problem lies in AW[$*$].
We will show \textnormal{AW[$*$]}-hardness by reducing from the model-checking problem on the class of all graphs.
Let $G$ be a graph and $\phi$ be a first-order sentence.
We will define a \FOX sentence $\hat\phi$ and a tree $T$ of depth 4
such that  $G \models \phi$ iff $T \models \hat\phi$.
We can assume without loss of generality that $V(G) = \{1,\dots,n\}$ for some $n \ge 1$.
The tree $T$ is constructed by the following polynomial-time procedure:
\begin{itemize}
    \item insert a root vertex $r$,
    \item for $i \in V(G)$ insert vertices $a_i$ and edges $a_ir$,
    \item for $i \in V(G)$, $k \in \{1,\dots,n-i+2\}$ add vertices $b_{i,k},c_{i,k}$ and edges $c_{i,k}b_{i,k}$ and $b_{i,k}a_i$,
    \item for $ij \in E(G)$ add vertices $d_{i,j}$ and edges $d_{i,j}a_i$,
    \item for $ij \in E(G)$, $k \in \{1,\dots,i+1\}$ add vertices $e_{i,j,k}$ and edges $e_{i,j,k}d_{i,j}$.
\end{itemize}
The tree defined in~\cite[Theorem 4.1]{GroheS2018}
has the same underlying structure as ours.
We therefore refer to~\cite{GroheS2018}
for the construction of auxiliary formulas $\phi_a(x),\dots,\phi_e(x)$
that identify the set of $a,\dots,e$-vertices.
They are satisfied in $T$ if and only if the free variable $x$ is assigned to an $a,\dots,e$-vertex, respectively.
For every vertex $i$ in $G$ we know that $a_i$ is the unique $a$-vertex in $T$ with exactly $(n-i+2)$ $b$-neighbors.
This yields a bijection between the vertices of $G$ and the $a$-vertices of $T$.
Also a vertex $j$ is adjacent to $i$ if and only if $a_j$ has a $d$-neighbor with exactly $(i+1)$ $e$-neighbors.
If we combine these observations we see that $i$ and $j$ are adjacent in $G$
if and only if $a_j$ has a $d$-neighbor $d$ such that
the number of $e$-neighbors of $d$ plus the number of $b$-neighbors of $a_i$ equals $n+3$.
The following \FOX-formula expresses this property (note that $=$ can
be simulated by~$>$ in this context):
$$
\psi_E(a,a') = \exists d 
\bigl( E(a,d) \wedge 
    \bigl( \cnt{y} (E(d,y) \wedge \phi_e(y)) \vee (E(a,y) \wedge \phi_b(y)) \bigr) = n+3
\bigr)
$$
Therefore $G \models E(i,j)$ iff $T \models \psi_E(a_i,a_j)$.
We construct $\hat\phi$ from $\phi$ by replacing each occurrence of $E(x,y)$ with $\psi_E(x,y)$
and restricting all quantifiers to $a$-vertices (i.e., replacing subformulas
$\exists x \psi$ by $\exists x (\phi_a(x) \wedge \psi)$).
It holds that $G \models \phi$ iff $T \models \hat\phi$.
\end{proof}

\subsection{Hardness of Approximation for Extensions of \BOLDFOX}\label{sec:foxextensionhard}

In preparation for \Cref{lem:approxhard},
we represent the edge relationship of an arbitrary graph
by a binary encoding that needs only logarithmic number of bits per edge.
This means that 
model-checking is already hard on bipartite
graphs where one side of the graph has only logarithmic size.
In this work, $\log()$ is the logarithm to the base two.

\begin{lemma}\label{lem:bipartite}
The model-checking problem for \textnormal{FO} on
the class of all
bipartite graphs with sides $U$, $V$ such that $|U| = \log(|V|)$
is \textnormal{AW[$*$]}-complete.
\end{lemma}

\begin{proof}
This proof will use some elements of the proof of \Cref{lem:exacthard}.
We will again show \textnormal{AW[$*$]}-hardness by reducing from the model-checking problem on the class of all graphs.
Let $G$ be a graph and $\phi$ be a first-order sentence.
We will define an FO sentence $\hat\phi$ and bipartite graph $\hat G$
such that  $G \models \phi$ iff $\hat G \models \hat\phi$.
It will be convenient for us to have a binary representation of each vertex of $G$ with at least three ones.
We therefore assume without loss of generality that
$V(G) = \{\,8i-1 \mid i \in \{1,\dots,n\}\,\}$ for some $n \ge 1$.
We first define the two sides $U$ and $V$ of the vertex set of our bipartite graph $\hat G$.
The node set $U$ contains
\begin{itemize}
    \item vertices $c_k$, $d_k$ for $k \in \{1,\dots,\floor{\log(8n)}+1\}$,
    \item vertices $e_1$, $e_3$, $e_5$.
\end{itemize}
The node set $V$ contains
\begin{itemize}
    \item vertices $a_i$ for $i \in V(G)$,
    \item vertices $b_{i,j}$ for $ij \in E(G)$,
    \item vertices $e_4$, $e_2$,
    \item $|G|^{O(1)}$ many ``padding'' vertices until $|V| = 2^{|U|}$.
\end{itemize}
The $e$-vertices are there only for technical reasons.
The $a$- and $b$-vertices will represent the vertices and edges of $G$, respectively.
But we are not allowed to connect them directly in $\hat G$.
Instead, we define a logarithmic number of $c$- and $d$-vertices
and connect each $a$- and $b$-vertex in a unique way to the $c$- and $d$-vertices.
A vertex $a_i$ or $b_{i,j}$ will be connected to the $c$- and $d$-vertices
according to a binary encoding of $i$ and $j$.
We will then define a formula $\hat\phi$ that uses this encoding to test for two $a$-vertices whether
there is a $b$-vertex representing an edge between them.
We add exactly the following edges to $\hat G$:
\begin{itemize}
    \item edges $a_ic_k$, $a_id_k$ for $i \in V(G)$ and 
        every $k \in \{1,\dots,\floor{\log(8n)}+1\}$ 
        such that the $k$th bit of the binary encoding of $i$ is one,
    \item edges $b_{i,j}c_k$ for $ij \in E(G)$ and 
        every $k \in \{1,\dots,\floor{\log(8n)}+1\}$ 
        such that the $k$th bit of the binary encoding of $i$ is one,
    \item edges $b_{i,j}d_k$ for $ij \in E(G)$ and 
        every $k \in \{1,\dots,\floor{\log(8n)}+1\}$ 
        such that the $k$th bit of the binary encoding of $j$ is one,
    \item edges $e_1a_i$ for $i \in V(G)$,
    \item edges $e_4c_k$ for $k \in \{1,\dots,\floor{8\log(n)}+1\}$,
    \item edges $e_1e_2$, $e_2e_3$, $e_3e_4$, $e_4e_5$.
\end{itemize}
The graph $\hat G$ can be constructed in polynomial time from $G$.
Next, we define formulas $\phi_a(x),\phi_b(x),\phi_c(x)$, $\phi_d(x)$ to
identify the $a,b,c,d$-vertices, respectively:
Since the binary encoding of every vertex in $G$ has at least three ones,
all $a,b,c,d$-vertices have degree at least three.
Therefore, $e_5$ is the unique vertex with degree one and $e_4$ is its only neighbor.
Also, $e_3$ and $e_2$ are the unique vertices with degree two and distance two and three from $e_5$, respectively.
Furthermore, $e_1$ is the unique neighbor of $e_2$ which is not $e_3$.
The $a$-vertices are the neighbors of $e_1$ with degree at least three.
The $c$-vertices are the neighbors of $e_4$ with degree at least three.
The $b$- and $d$-vertices are the remaining vertices with distance one and two from $a$-vertices, respectively.

Two vertices $i$ and $j$ are adjacent in $G$
if and only if there exists a $b$-vertex
with the same $c$-neighbors as $a_i$ and the same $d$-neighbors as $a_j$.
This is checked by the formula
\begin{multline*}
\psi_E(a,a') = \exists b 
\bigl( \phi_b(b) \wedge
        \bigl(  \forall c\, \phi_c(c) \rightarrow (E(c,a) \leftrightarrow E(c,b)) \bigr)
        \wedge \\
        \bigl(  \forall d\, \phi_d(d) \rightarrow (E(d,a') \leftrightarrow E(d,b)) \bigr)
\bigr).
\end{multline*}
Then $G \models E(i,j)$ iff $\hat G \models \psi_E(a_i,a_j)$.
Just like in \Cref{lem:exacthard},
we construct $\hat\phi$ from $\phi$ by replaying each occurrence of $E(x,y)$ with $\psi_E(x,y)$
and restricting all quantifiers to $a$-vertices.
We have $G \models \phi$ iff $\hat G \models \hat\phi$.
\end{proof}

We now show that certain fragments of \FOCless 
similar to \FOX do not admit an fpt
model-checking approximation scheme on the class of all trees with depth 9.
Using \Cref{lem:bipartite},
we only need to encode the edge relationship of a bipartite graph
where one side has only logarithmic size.
By again using a binary encoding,
we need roughly $\log(\log(n))$ bits to identify a vertex
on the small side.
This very small number of bits allows
our encoding to be stable in the approximation setting.

We do not show hardness for formulas constructed by the rule
$\cnt{y}\phi_1 + \cnt{z}\phi_2 > N$.
In fact, with some additional effort, our
approximation scheme can be extended to formulas of this form as well.
We refrained from doing so
because it does not offer much more additional expressive power.

Note that the multiplication with one in the term $1 \cdot \#y\,\phi_1$ of rule 1 seems redundant,
but is needed to comply with our definition of $\lambda$-similarity.
In a $\lambda$-similar formula the constant $1$ can be replaced
with an arbitrary number between $1$ and $\lambda$.

\begin{lemma}\label{lem:approxhard}
Let $\textnormal{L}$ be a fragment of \FOCless obtained by extending \textnormal{FO} by one
of the following four rules (with semantics as expected, see \cite{KuskeS2017}
for a rigorous definition):
Let $y$ and $z$ stand for arbitrary variables.
\begin{enumerate}
    \item
        If $\phi_1$, $\phi_2$ are formulas,
        then $1 \cdot \cnt{y}\phi_1 > 1 \cdot \cnt{z}\phi_2$ is a formula.
    \item
        If $\phi_1$, $\phi_2$ are formulas
        and $N \in \Z$,
        then $\cnt{y}\phi_1 - \cnt{z}\phi_2 > N$ is a formula.
    \item
        If $\phi_1$, $\phi_2$ are formulas
        and $N \in \Z$,
        then $\cnt{y}\phi_1 \cdot \cnt{z}\phi_2 > N$ is a formula.
    \item
        If $\phi$, is a formula
        and $N \in \Z$,
        then $\cnt{yz}\phi > N$ is a formula.
\end{enumerate}
Unless \textnormal{AW[$*$]} $\subseteq$ \textnormal{FPT},
there is no algorithm with the following properties:
It gets as input a sentence $\phi \in \textnormal{L}$, a tree $T$ of depth $9$,
runs in time at most $f(|\phi|)|T|^c$ for some function $f$ and constant $c$
and returns either $1$, $0$, or $\bot$.
\begin{itemize}
\item
    If the algorithm returns $1$ then $T \models \phi$.
\item
    If the algorithm returns $0$ then $T \not\models \phi$.
\item
    If the algorithm returns $\bot$ then $\phi$ is $\approxratio{|T|}$-unstable on $T$.
\end{itemize}
\end{lemma}

\begin{proof}
Since
$\cnt{y}\phi_1 \cdot \cnt{z}\phi_2 > N$ is equivalent
to $\cnt{yz}\phi_1 \land \phi_2 > N$ after possibly renaming variables,
rule 3 is more expressive than rule 4.
Also $1 \cdot \cnt{y}\phi_1 > 1 \cdot \cnt{z}\phi_2$ is equivalent to
$\cnt{y}\phi_1 - \cnt{z}\phi_2 > 0$ and therefore
rule 2 is more expressive than rule 1.
Thus from now on, we only have to consider rule 1 and 3.

For each rule, we will provide a polynomial-time procedure which
takes as input a bipartite graph $G$ 
with sides $U$, $V$ as well as $|U| = \log(|V|)$
and a FO sentence $\chi$
and constructs a tree $T$ and an L-sentence $\phi$ 
that is $\approxratio{|T|}$-stable on $T$
and whose size is bounded by a function of $\chi$,
such that $G \models \chi$ if and only if $T \models \phi$.
Since $\phi$ is $\approxratio{|T|}$-stable on $T$,
a model-checking approximation algorithm for L
as presented in this lemma would never yield $\bot$ on input $(T,\phi)$
and therefore decides in fpt time whether $T \models \phi$.
This would then decide whether $G \models \chi$,
which by Lemma~\ref{lem:bipartite} implies 
AW[$*$] $\subseteq$ FPT.

Let $\chi$ be a first-order sentence and
$G$ be a bipartite graph with sides $U$, $V$ and $|U| = \log(|V|)$.
We can assume without loss of generality that $U=\{1,\dots,\log(n)\}$, $V=\{1,\dots,n\}$ for some sufficiently large $n$.
We will first define some constructions that are
independent of whether L is defined using rule 1 or 3 and
later distinguish between the two cases.
Let $\lambda = \floor{2^{\log(n)^{2/3}}}$.
The tree $T$ will contain different ``gadget'' trees as subtrees.
We construct a ``gadget'' tree $T_i$, for $i \in \{1,\dots,\log(n)\}$ using the following procedure:
We start with inserting a root $q$.
Then for every ${l \in \{1,\dots,\floor{\log(\log(n))}+1\}}$ we
\begin{itemize}
    \item insert a vertex $a_l$ and an edge $qa_l$,
    \item if the $l$th bit of the binary encoding of $i$ is one insert a vertex $b_l$ and an edge $a_lb_l$,

    \item insert $\lambda^{3l}$ many so-called ``$c$-vertices'', 
        for each new $c$-vertex $c$ also add a new vertex $c'$ and edges $a_lc$, $cc'$,

    \item insert $\lambda^{4\floor{\log(\log(n))} - 3l}$ many so-called ``$d$-vertices'', 
        for each new $d$-vertex $d$ also add new vertices $d'$, $d''$ and edges $a_ld$, $dd'$, $d'd''$,

    \item insert $\lambda^{3l+1}$ many so-called ``$c^+$-vertices'', 
        for each new $c^+$-vertex $c$ also add new vertices $c'$, $c'', c'''$ and edges $a_lc$, $cc'$, $c'c''$, $c''c'''$,
    \item insert $\lambda^{3l-1}$ many so-called ``$c^-$-vertices'', 
        for each new $c^-$-vertex $c$ also add new vertices $c'$, $c'', c''', c''''$ and edges $a_lc$, $cc'$, $c'c''$, $c''c'''$, $c'''c''''$.
\end{itemize}
The number of $c$-, $c^+$-, $c^-$- and $d$-children are all powers of $\lambda$.
This will later help us build $\lambda$-stable L-formulas that can identify
the different $a$-vertices based on their number of children.
Since $\lambda^{\floor{\log(\log(n))}} 
= n^{O(1)}$, each gadget tree $T_i$ has polynomial size
and can be constructed in polynomial time.
We construct $T$ of depth $9$ with the following polynomial-time procedure:
\begin{itemize}
    \item insert a root $r$,
    \item for $i \in U$ insert a copy of $T_i$, rename its root $u_i$ and add an edge $u_ir$,
    \item for $j \in V$ insert vertices $v_j$, $e_j$, $f_j$ and edges $rv_j$, $v_je_j$, $v_jf_j$,
    \item for $ij \in E(G)$ with $i \in U$ and $j \in V$ insert a copy of $T_i$, name its root $w_{i,j}$ and add an edge $w_{i,j}v_j$.
\end{itemize}

In order to build $\phi$, we first need 
formulas $\phi_a(x)$, $\phi_b(x)$, $\phi_c(x)$, $\phi_{c^+}(x)$,
$\phi_{c^-}(x)$, $\phi_d(x)$, $\phi_u(x)$, $\phi_v(x)$, $\phi_w(x)$ defining
the set of corresponding vertices.
They can be build using these observations:
\begin{itemize}
    \item $a$-vertices are those vertices with at least three neighbors of degree two,
    \item the $b$-, $c$-, $d$-, $c^+$-, $c^-$-vertices are the neighbors of $a$-vertices
        with degree one or two and ``tails'' of length 0, 1, 2, 3, 4 respectively,
    \item $w$-vertices are those with more than one $a$-neighbor and distance two to a leaf,
    \item $u$-vertices are those with more than one $a$-neighbor and not distance two to a leaf,
    \item $v$-vertices are those adjacent to a $w$-vertex and exactly two leaves.
\end{itemize}

There is a bijection between $U$ and the $u$-vertices in $T$ and $V$ and the $v$-vertices in $T$.
We want to check whether two vertices $i \in U$ and $j \in V$
are adjacent in $G$ by evaluating a formula with two free variables on $u_i$ and $v_j$ in $T$.
We know that $u_i$ is the root of a gadget tree isomorphic to $T_i$.
Furthermore $i$ and $j$ are adjacent in $G$ if and only if $v_j$ has a $w$-child that is the root of a gadget tree isomorphic to $T_i$.
In order to test that, we need an L-formula $\psi_\text{gadget}(x,y)$ such that
for every $u$-vertex $u$ and $w$-vertex $w$ in $T$ holds 
$T \models \psi_\text{gadget}(u,w)$ if and only if
$u$ and $w$ are the roots of two isomorphic gadget trees.
We further require $\psi_\text{gadget}(u,w)$ to be $\lambda$-stable on $T$.
If we have such a formula, the remaining proof is
similar to the ones of \Cref{lem:exacthard}~and~\ref{lem:bipartite}.

Let $u,w$ be the roots of two gadget trees isomorphic to $T_i,T_{i'}$ in $T$ respectively.
For each $l \in \{1,\dots,\floor{\log(\log(n))}+1\}$ let
$a^u_l$ and $a^w_l$ be the $a$-child of $u$ and $w$, respectively, with $\lambda^{3l}$ $c$-children.
We first construct a formula $\psi_\text{pair}(x,y)$ such that
$T \models \psi_\text{pair}(a^u_l,a^w_{l'})$ iff $l = l'$.
For this, we have to distinguish whether L extends FO using rule 1 or 3.

\begin{itemize}
\item[1.]
This rule can compare two non-atomic counting terms.
With this ability, we can pair $a^u_l$ with $a^w_{l'}$
if the number of $c$-children of $a^u_l$ lies
between the number of $c^-$-~and $c^+$-children of $a^w_{l'}$.
The pairing is achieved by the formula
\begin{align*}
    \psi_\text{pair}(x,y) = 
           &\; \bigl( 1 \cdot \cnt{z} \phi_{c^+}(z) \wedge E(z,x) \ge 1 \cdot \cnt{z} \phi_c(z) \wedge E(z,y) \bigr) \\ 
    \wedge &\; \bigl( 1 \cdot \cnt{z} \phi_{c^-}(z) \wedge E(z,x) \le 1 \cdot \cnt{z} \phi_c(z) \wedge E(z,y) \bigr).
\end{align*}

\item[3.]
This rule can compare the product of two non-constant counting terms.
With this ability, we can pair $a^u_l$ with $a^w_{l'}$
if the product of the number $c$-children of $a^u_l$ and
the the number of $d$-children of $a^w_{l'}$ equals $\lambda^{4\floor{\log(\log(n))}}$.
This pairing can be modeled by the formula
\begin{align*}
    \psi_\text{pair}(x,y) = 
           & \bigl( \cnt{z} \phi_c(z) \wedge E(z,x) \cdot \cnt{z} \phi_d(z) \wedge E(z,y) \ge \lambda^{4\floor{\log(\log(n))} - 1} \bigr) \\ 
    \wedge & \bigl( \cnt{z} \phi_c(z) \wedge E(z,x) \cdot \cnt{z} \phi_d(z) \wedge E(z,y) \le \lambda^{4\floor{\log(\log(n))} + 1} \bigr).
\end{align*}
\end{itemize}
For two different $a$-children, their number of
$c$-, $c^+$-, $c^-$-, $d$-children differ at least by a factor of $\lambda^3$.
Therefore, we still have
$T \models \psi'_\text{pair}(a^u_l,a^w_{l'})$ iff $l = l'$
for every $\lambda$-similar formula $\psi'_\text{pair}$.
This makes $\psi_\text{pair}(a^u_l,a^w_{l'})$ $\lambda$-stable on $T$.
We can now compare the binary encoding of $i$ and~$i'$:
It holds $i=i'$ if and only if for each pair of $a$-vertices either both or none have a $b$-child.
We define
\begin{multline*}
    \psi_{\textnormal{gadget}}(x,y) = \forall a \forall a' 
    \Bigl( 
        \bigl(
             \phi_a(a) \land \phi_a(a') \land E(a,x) \wedge E(a',y) \land \psi_\text{pair}(a,a') 
        \bigr)
    \to \\
        \bigl( 
            \bigl(\exists b \phi_b(b) \wedge E(a,b)\bigr) \leftrightarrow \bigl(\exists b \phi_b(b) \wedge E(a',b)\bigr)
        \bigr)
    \Bigr).
\end{multline*}
Now $i=i'$ iff $T \models \psi_\textnormal{gadget}(u,w)$.
Since $\psi_\textnormal{pair}(a^u_l,a^w_{l'})$ is $\lambda$-stable,
$\psi_\textnormal{gadget}(u,w)$ also is $\lambda$-stable on $T$.
Using
$$
\psi_E(x,y) = \exists z \bigl( \phi_w(z) \wedge E(y,z) \wedge \psi_{\text{gadget}}(x,z) \bigr),
$$
we have for every $i \in U$, $j \in V$ that
$i$ and $j$ are adjacent in $G$ iff $T \models \psi_E(u_i,v_j)$.
We can then construct $\phi$ from $\chi$
by replacing every occurrence of $E(x,y)$ with
$$
\bigl(\phi_u(x) \land \phi_v(y) \land \psi_E(x,y)\bigr) \lor
\bigl(\phi_u(y) \land \phi_v(x) \land \psi_E(y,x)\bigr)
$$
and relativizing all quantifiers to $u$- and $v$-vertices by replacing subformulas
$\exists x \psi$ with $\exists x ((\phi_u(x) \vee \phi_v(x)) \wedge \psi)$.
Then $G \models \chi$ iff $T \models \phi$.
Since $\psi_\text{gadget}$ is $\lambda$-stable on $T$ for all $u$- and $w$-vertices,
$\phi$ also is $\lambda$-stable on~$T$.
The tree $T$ has polynomial size and therefore
$$
\approxratio{|T|} = 2^{\sqrt{O(1)\log(n)}} = o(2^{\log(n)^{2/3}}) = o(\lambda).
$$
Without loss of generality we can assume $n$ to be sufficiently large that
$\approxratio{|T|} \le \lambda$.
This means that $\phi$ is $\approxratio{|T|}$-stable on $T$.
\end{proof}

\section{Open Questions}

We see the following open questions, 
sorted in descending order by estimated difficulty.

It should be possible to 
generalize our \FOX model-checking approximation scheme from 
bounded expansion graph classes to nowhere dense graph classes.
However, our approach using functional representations and quantifier elimination
will most likely not be applicable.
Since a \FOX model-checking approximation scheme in particular
also solves the model-checking problem for FO, we cannot
hope to extend our results beyond nowhere dense graph classes
(assuming monotonicity)~\cite{GroheKS2017}.

Can our optimization result in \Cref{thm:optimization}
for counting-terms of the form $\#y\phi(y\bx)$
with $\phi(y\bx) \in$ FO be extended to nowhere dense graph classes?
At this point, we do not even know whether we can efficiently solve
partial dominating set (special case of \Cref{thm:optimization})
on nowhere dense graph classes.
We believe both to be the case.
There is potential for further optimization results
similar to \Cref{thm:optimization}
for different fragments of \FOC 
on bounded expansion or possibly nowhere dense graph classes.

Can our model-checking approximation scheme be generalized to query-counting or query-enumeration?
For every \FOX-formula $\phi(\bx)$, structure $G$ and $\varepsilon > 0$ we define:
\begin{align*}
    \lceil\phi,G,\varepsilon\rceil   &= \{ \bu \in V(G)^{|\bx|} \mid \text{there exists $\phi'$ that is $(1+\varepsilon)$-similar to $\phi$ with $G \models \phi'(\bu)$} \} \\
    \lfloor\phi,G,\varepsilon\rfloor &= \{ \bu \in V(G)^{|\bx|} \mid \text{for all $\phi'$ that are $(1+\varepsilon)$-similar to $\phi$ holds $G \models \phi'(\bu)$} \}.
\end{align*}
Can we compute in time $f(|\phi|,\varepsilon)\norm{G}$ a number
$s \in \N$ with
$|\lfloor\phi,G,\varepsilon\rfloor| \le s \le |\lceil\phi,G,\varepsilon\rceil|$?
Is it possible to enumerate with constant-delay
a set $S$ with $\lfloor\phi,G,\varepsilon\rfloor \subseteq S \subseteq \lceil\phi,G,\varepsilon\rceil$?
Using our approximate quantifier elimination procedure in \Cref{sec:iteratedElimination}, one probably can
replace $\phi(\bx)$ with a quantifier-free formula and then
use existing query-counting and query-enumeration techniques.

Consider the logic FOC$(\{=_p\})$, where $=_p$ is the equality relation modulo $p$.
Then, for example, the FOC$(\{=_2\})$-formula $\forall x \, \# y \, E(x,y) =_2 0$ expresses
whether a graph has an Euler cycle.
It is already known that the model-checking problem for FOC$(\{=_p\})$
is fpt on graph classes with bounded degree~\cite{HeimbergKS2016}.
One can most likely solve it also on graph classes with bounded expansion
with a quantifier elimination procedure based on \Cref{lem:apxExact} and unary predicates as in \Cref{lem:2fo}.
For nowhere dense graph classes the proof might be considerably more difficult
and require different techniques.

\bibliographystyle{plainurl}
\bibliography{counting}

\end{document}